\documentclass[twocolumn,aps,pra,showpacs,showkeys,groupedaddress]{revtex4-1}
\usepackage{graphicx}
\usepackage{amsthm}
\usepackage{epsfig,amssymb,amsmath,mathrsfs,color}
\usepackage{epsfig}
\usepackage[linkcolor=red]{hyperref}
\usepackage[matrix,frame,arrow]{xypic}
\usepackage{mathtools}
\usepackage{dsfont}
\usepackage{comment}
\usepackage[font=small,labelfont=bf,
   justification=justified,
   format=plain]{caption}
\usepackage[font=small,labelfont=bf,
justification=justified,
format=plain]{subcaption}
\renewcommand{\v}[1]{\ensuremath{\mathbf{#1}}}
 
\newcommand{\ket}[1]{\left| #1 \right>}
\newcommand{\bra}[1]{\left< #1 \right|}

 \newtheorem{remark}{Remark} 
\newtheorem{proposition}{Proposition}

\newtheorem*{example*}{Example}
\newtheorem{corollary}{Corollary} 
 \newtheorem{definition}{Definition}
\def\bvec#1{\mathbf{#1}}
\def\bk{\bvec k}
\def\bh{\bvec h}
\def\bik{{\bvec  k'}}
\def\bn{\bvec n}

\def\Proof{\medskip\par\noindent{\bf Proof.  }}
\def\qed{$\,\blacksquare$\par} \def\>{\rangle} \def\<{\langle}

\def\Tr{{\rm Tr}}  \def\Reals{{\mathbb R}}

\usepackage[usenames,dvipsnames]{xcolor}

\usepackage{soul}

\begin{document}

\title{Isotropic quantum walks on lattices and the Weyl equation}

\author{Giacomo Mauro D'Ariano}\email{dariano@unipv.it}
\author{Marco Erba}\email{marco.erba01@ateneopv.it}
\author{Paolo Perinotti}\email{paolo.perinotti@unipv.it}
\affiliation{Universit\`a degli Studi di Pavia, Dipartimento di Fisica, QUIT Group, and INFN Gruppo IV, Sezione di Pavia, via Bassi 6, 27100 Pavia, Italy}

%\date{\today}
\begin{abstract}
We present a thorough classification of the isotropic quantum walks on lattices of dimension $d=1,2,3$ with a coin system of dimension $s=2$. For $d=3$ there exist two isotropic walks, namely the Weyl quantum walks presented in Ref.~\cite{DP14}, resulting in the derivation of the Weyl equation from informational principles. The present analysis, via a crucial use of isotropy, is significantly shorter and avoids a superfluous technical assumption, making the result completely general.
\end{abstract}
\maketitle

\section{Introduction}\label{sec:intro}
Recently the possibility of implementing actual quantum simulations of quantum fields~\cite{1367-2630-14-7-073050,PhysRevA.88.023617,1367-2630-19-6-063038,PhysRevLett.112.120406} has been accompanied by novel approaches to foundations of the theory~\cite{PhysRevA.94.012335,ArrighiNJP14,Arrighi2016,PhysRevA.95.022106}, including its derivation from informational principles~\cite{DP14,BDP16} and the recovery of its Lorentz covariance~\cite{PhysRevA.94.042120}. This has provided a progress in the research based on the idea originally proposed by Feynman~\cite{Feynman1982} of recovering physics as pure quantum information processing. Deriving quantum field theory from just denumerable quantum systems provides an emergent notion of space-time, with no prior background. This suggests that the approach may be promising for a future development of quantum theories of gravity.

The mathematical formalisation of the discrete quantum algorithm running a quantum field dynamics is provided by the notion of quantum cellular automaton~\cite{watrous1995one,schumacher2004reversible,ARRIGHI2011372}. A quantum cellular automaton is a unitary homogeneous evolution of the algebra of local observables that preserves locality. When the automaton is linear in the local algebra generators, the cellular automaton is usually referred to as a quantum walk (QW)~\cite{ADZ93,AB01,S03}, and is suited for the description of the free field theory for a fixed number of particles. 

A quantum walk on a graph represents a coherent counterpart of a classical random walk on the same graph. In the derivation of Ref.~\cite{DP14} it was proved that, if one assumes homogeneity of the evolution, the graph must be the {\em Cayley graph} of a group $G$. When the graph corresponds to a free Abelian group $G\cong\mathbb{Z}^d$, one finds the two \emph{Weyl QWs} (one for the left- and one for the right-handed mode), recovering the Weyl equation in $d+1$ dimensions for $d=1,2,3$. An alternative derivation of the Weyl QWs for $d=3$ on the BCC lattice has been recently presented in Ref.~\cite{PhysRevA.95.062344}. In Ref.~\cite{DP14} the derivation of the Weyl QWs exploited the technical assumption that there is a quasi-isometry~\cite{JM08} of the Cayley graph in a Euclidean manifold such that no vertex can lie within the sphere of nearest neighbours. On the other hand, most of the derivation did not use the isotropy principle. In the present paper, on the contrary, we exploit the isotropy principle from the very beginning of the derivation, thus avoiding the above assumption and making the classification of the isotropic QWs on $\mathbb{Z}^d$ completely general. In the present paper the derivation of the Weyl QWs is included in a complete classification of isotropic QWs on lattices of dimension $d = 1, 2, 3$ with a coin system of dimension $s = 2$. The result exploits the isotropy notion of Ref. \cite{DP14}, which is extended in this paper in order to account for groups with generators of different orders. We will introduce a technique to construct the Cayley graphs of a given group $G$ supporting an isotropic QW. Remarkably, the Cayley graph is unique for each dimension $d=1,2,3$.

The manuscript is organized as follows. In Sec.~\ref{sec:QW} we review the notion of Cayley graph of a group $G$, and define QWs on Cayley graphs, introducing the definition of isotropy and its main properties. In Sec.~\ref{sec:Zd} we review the theory of QWs on free Abelian groups. In Sec.~\ref{sec:presentations} we select the possible Cayley graphs according to a necessary condition for a QW to be isotropic. In Sec.~\ref{Sec:Derivation} we prove a second necessary condition for isotropy that is used in the appendix to refine the selection of Cayley graphs, and we solve the unitarity condition on the selected Cayley graphs for $d=1,2,3$, finding the two Weyl QWs. Sec.~\ref{Sec:conclu} closes the paper with some concluding remarks, whereas in Appendix~\ref{app:excluding} we report technical proofs and details.

\section{Isotropic QWs on Cayley graphs}\label{sec:QW}

We now define the QW on a Cayley graph $\Gamma(G,S_+)$ of a group $G$, with generating set $S_+$. A generating set $S_+\subseteq G$ is a set of elements of $G$ such that all the elements of the group can be expressed as words of elements of $S_+$ along with their inverses. The Cayley graph is a coloured directed graph with the elements of $G$ as vertices and the elements of $S_+$ as edges: a colour is associated to each generator $h\in S_+$, and two vertices $g,g'\in G$ are connected by the coloured edge $h\in S_+$ if $g'=gh$, with the arrow directed from $g$ to $g'$. In the following we will take $|S_+|<\infty$, namely the group $G$ is finitely generated. The Cayley graph of a group can be defined by giving a \emph{presentation}, namely choosing a set of generators (an alphabet) and a set of \emph{relators}, i.e~a set of words which are equal to the identity of $G$. This completely specifies a unique group $G$. The cardinality of the group $G$ can be finite or infinite, depending on its relators, however the most interesting case in the present context is that of a finitely presented infinite
group.

Let $\{|g\>\}_{g\in G}$ be an orthonormal basis for $\ell^2(G)$. The right-regular representation $T$ of $G$ is defined as
\begin{equation}
T_g|g'\>:=|g'g^{-1}\>.
\end{equation}
A QW on the Cayley graph $\Gamma(G,S_+)$ of the group $G$ is a unitary operator $A$ on $\ell^2(G)\otimes \mathbb{C}^s$, with $1\leq s<\infty$, that can be written as
\begin{equation*}
A=\sum_{h\in S} T_h\otimes A_h,
\end{equation*}
where $S=S_+\cup S_-$, $S_-=S_+^{-1}$ is the set of inverses of $S_+$, and $\{A_h\}_{h\in S}\subseteq\mathbb{M}_s(\mathbb{C})$ are the so-called \emph{transition matrices} of the QW.

It is worth mentioning that also other constructions of QWs have been given in the literature, for example QWs such that the coin system is generated by the set of edges of the underlying graph (see e.g.~Ref~\cite{montanaro2007quantum}, and Ref.~\cite{kempe2003quantum} for an overview).

Generally we will consider also self-transitions, corresponding to the inclusion of the identity $e\in G$ in the generating set which is then given by $S\equiv S_+\cup S_-\cup\{e\}$. In the following, for each group $G$ considered, we will assume $A_h\neq 0$ for all $h\in S_+\cup S_-$, whereas in general we allow for the case $A_e=0$.
We also denote by $S^{n}_+ \subseteq S_+$ the set of  generators of order $n\geq 2$, i.e.~$n$ is the smallest integer such that $h^n = e$. Notice that the most common case is that of $n=+\infty$.

For the purpose of introducing the concept of isotropic QWs, we remind that a graph automorphism is defined as a bijective map of the vertices that preserves the set of edges. For a Cayley graph this means that the automorphism $l$ is such that if $g'=gh$, then $l(g')=l(g)h'$, with $g,g'\in G$ and $h,h'\in S_+$. Then, an automorphism of the Cayley graph can be expressed as a permutation $\lambda$ of the set of colours $S_+$, where for every $g\in G$ and $h\in S_+$ one has $l(gh)=l(g)\lambda(h)$ for some permutation $\lambda$ of $S_+$. Let us denote by $\Lambda$ a group of permutations of the elements of $S_+$.
\begin{definition}[Isotropic QW]
A QW on $\Gamma(G,S_+)$ is called \emph{isotropic} with respect to $S_+$ if there exists a group $L$ of automorphisms of $\Gamma(G,S_+)$ that can be expressed as a permutation of the colours $S_+$, such that 
the evolution operator of the QW is $L$-covariant, i.e.~there exists a projective unitary representation $U$ over $\mathbb{C}^s$ of $L$ such that 
	\begin{equation*}
	A_{\lambda(h)} = U_l A_h U_l^\dagger\quad \forall l\in L,\forall h\in S_+,
	\end{equation*}
where $\lambda\in\Lambda$, and such that the action of $\Lambda$ is transitive on each subset $S_+^n$. 
\end{definition}
The previous definition guarantees that the group of local changes of basis representing the isotropy group $L$---which is a group of automorphisms of the graph---acts just as a permutation of the transition matrices, implying that all the directions are dynamically equivalent.

To satisfy homogeneity, one has to demand also the following condition~\footnote{The homogeneity requirement defined in Ref.~\cite{DP14} should be completed upon requiring that any two nodes remain ``distinguishable'' from the point of view of a third node. For details we will refer to Ref.~\cite{unpubDP17}. Eq.~\eqref{comm} follows from this definition.}:
\begin{equation}\label{comm}
[U_l,A_h]\neq 0\quad \forall h\in S_+,\forall l\in L : l(h)\neq h .
\end{equation}
Indeed, two transition matrices associated to different generators must be distinct. 
In particular, this implies that if $L$ does not contain nontrivial elements stabilizing all the $h\in S$, then the representation $U$ must be faithful (otherwise it would contain at least one nontrivial element represented as $I_s$).

\begin{proposition}\label{Pr1} The automorphisms of the Cayley graph $\Gamma(G,S_+)$ are also automorphisms of  $G$.
\end{proposition}
\begin{proof}
Consider the action of arbitrary elements $l \in L$ on the graph vertices. We have 
\begin{equation*}
\lambda(h) = l(h)=l(eh) = e\lambda(h),\quad \forall h\in S_+,
\end{equation*}
and since $l(gh)=l(g)\lambda(h)$ $\forall g\in G$, then $l(e)=e$. The same holds $\forall h\in S_-$.  Moreover
\begin{equation*}
l(hh') = l(h)\lambda(h') \equiv l(h)l(h')\quad \forall h,h'\in S.
\end{equation*}
Iterating, in general we obtain 
\begin{equation}\label{finiteL}
l(h_{1}\cdots h_{p}) = l(h_{1}) \cdots l(h_{p}),\quad \forall h_1,\ldots ,h_p \in S,
\end{equation}
and, being $S$ a set of generators for $G$, this amounts to
\begin{equation*}
 l(gg') = l(g)l(g')\quad \forall g,g'\in G.
\end{equation*}
Accordingly, $L$ is a group automorphism of $G$.
\end{proof}
The isotropy conditions corresponds to the covariance
\begin{equation}
A=\sum_{h\in S} T_h\otimes A_h=\sum_{h\in S} T_{l(h)}\otimes U_lA_hU^\dag_l\quad \forall l\in L.\label{covW}
\end{equation}
The covariance condition \eqref{covW} and the transitivity of $\Lambda$ on each $S^n_+$ imply, by linear independence of the $T_h$, that every $S_+^{n}$ is invariant under some subgroup $L^n \leq L$. In fact, any $S_+^n$ is the orbit of an arbitrary generator $h_1^{(n)}\in S_+^n$ under $L^n$, denoted with $\mathcal{O}_{L^n}(h_1^{(n)})$. 
\begin{proposition}\label{P1}
The isotropy group $L$ is a finite subgroup of $\operatorname{Aut}(G)$.
\end{proposition}
\begin{proof}
By Proposition~\ref{Pr1} the isotropy group $L$ is a group of automorphisms of $G$. By Eq. \eqref{finiteL} $L\cong\Lambda$, hence $L$ is finite.
\end{proof}
\begin{corollary}
Each subgroup $L^n \leq L$ is isomorphic to a finite permutation group acting transitively on $S_+^n$.
\end{corollary}
\begin{corollary}\label{cor2}
If all generators have the same order, $L$ is isomorphic to a finite permutation group acting transitively on $S_+$.
\end{corollary}
By Eq. \eqref{covW} one can always choose the projective unitary representation $U$ with unit determinant, namely
$U_l\in\mathbb{SU}(s)$ $\forall l\in L$. Notice that, by definition of isotropy, either $S_+^n$ does not contain the inverse of any of its elements or it coincides with the whole set $S^n \coloneqq S_+^n \cup S_-^n$.

%The above quantum cellular automaton corresponds to the description of
%a physical law by a quantum algorithm with finite algorithmic
%complexity, with homogeneity corresponding to the universality of the
%physical law. One can easily realize that the construction is very
%general, in consideration that the group $G$ is abstractly introduced
%via generators and relators, whence can include also random groups,
%and e.g. graphs for the automaton that are trees, and many different
%situations depending on the group. The whole physics will emerge
%without requiring any metrical structure, since the group is defined
%only topologically. 

In the following we will consider the isotropic QWs on $\Gamma(G,S_+)$ with $s=2$ and $G\cong \mathbb{Z}^d$ with $d=1,2,3$. For $d=3$ we discover that there are two QWs (modulo discrete symmetries) that for large-scales give the two Weyl equations, one for left- and one for right-handed mode. In Ref.~\citep{DP14} it is shown that, coupling two Weyl QWs in the only possible way consistent with the above requirements (specifically locality), the resulting QW is unique (modulo discrete symmetries) and describes exactly the Dirac equation for large scales.

\section{Quantum Walks on Cayley graphs of $\mathbb{Z}^d$}\label{sec:Zd}

Since we are considering Abelian groups, we will denote the group elements as usual with the
boldfaced vector notation as $\bn\in G$, and the generators as $\bh\in S$. Moreover, we will use the additive notation for the group composition, and $0$ for the identity element. The space $\ell^2(G)$ will be the span of $\{|\bn\>\}_{\bn\in G}$ and the generators $\bh$ are represented by the operators 
\begin{equation*}
T_\bh:=\sum_{\bn\in G}|\bn+\bh\>\<\bn|.
\end{equation*}
We now treat the elements of $G$ as vectors in $\Reals^d$.
Generally the elements of $S$ are linearly dependent. We introduce all the sets $D_n\subseteq S_+$ of linearly independent elements
\begin{equation*}
D_n:=\{\bh_{i_1},\ldots,\bh_{i_d}\}, 
\end{equation*}
where $n$ labels the specific subset. For every $D_n$ we construct the dual set $\tilde D_n$ defined by
\begin{equation*}
\tilde D_n:=\{\tilde \bh_1^{(n)},\ldots,\tilde \bh_d^{(n)}\},\quad
\end{equation*}
where 
\begin{equation*}
\tilde \bh^{(n)}_l\cdot \bh_{i_m}=\delta_{lm}.
\end{equation*}
Now we define the set 
\begin{equation*}
\tilde D:=\bigcup_n\tilde D_n.
\end{equation*}
The Brillouin zone $B\subseteq\Reals^d$ is defined as the polytope
\begin{equation*}
B=\bigcap_{\tilde\bh\in\tilde D}\{\bk\in\Reals^d\mid-\pi|\tilde\bh|^2\leq\bk\cdot\tilde\bh\leq\pi|\tilde\bh|^2\}.
\end{equation*}
The unitary operator of the QW is given by
\begin{equation}
  A=\sum_{\bn\in G}\sum_{\bh\in S}|\bn+\bh\>\<\bn|\otimes A_\bh.
  \label{eq:translinvpos}
\end{equation}
One has $[A,T_\bh\otimes I_s]=0$. The unitary irreps are one-dimensional,
and are classified by the joint eigenvectors of $T_\bh$ 
\begin{equation*}
T_{\bh_i}|\bk\>=:e^{-i\bk\cdot\bh_i}|\bk\>,
\end{equation*}
where 
%we number the elements $\bh_j$ of the generator set $S$. 
%Expanding $|\bk\>=\sum_{\bn\in
%  G}c_\bn(\bk)|\bn\>$, the last equation reads
%\begin{equation}
%  e^{ik_i}c_{\bn-\bh_i}(\bk)=c_{\bn}(\bk),
%\end{equation}
%namely
%\begin{align}
%  e^{i\bn\cdot\bk}c_{\bvec 0}(\bk)=c_\bn(\bk),
%  \quad\bn=\sum_{l=1}^dn_l\bh_{j_l},\quad\bk=\sum_{l=1}^dk_l\tilde{\bh}_{j_l},
%\end{align}
%where $\{\bh_{j_l}\}_{l=1}^d$ are linearly independent, and $\tilde{\bh}_{j_l}\cdot\bh_{j_m}=\delta_{lm}$. 
%Finally this implies
\begin{equation*}
  |\bk\>:=\frac1{\sqrt{|B|}}\sum_{\bn\in G}e^{i\bk\cdot\bn}|\bn\>,\quad|\bn\>=\frac1{\sqrt{|B|}}\int_Bd\bk e^{-i\bk\cdot\bn}|\bk\>.
\end{equation*}
%and $B$ is the Brillouin zone defined through the following set of linear constraints
%\begin{equation}
%  \begin{split}
%B:=&\{\bk|-\pi\leq \bk\cdot\tilde\bh_i\leq\pi\},
%  \end{split}
%  \label{eq:brill}
%\end{equation}
%where $\tilde\bh_i$ is any element of the dual set $\{\tilde\bh_i\}$
%for any choice of independent vectors set $\{\bh_i\}$. 
Notice that
\begin{equation*}
  \<\bk|\bik\>=\frac1{|B|}\sum_{\bn\in G}e^{i(\bk-\bik)\cdot\bn}=\delta_{2\pi}(\bk-\bik).
\end{equation*}
Translation invariance of the QW in Eq.~\eqref{eq:translinvpos} then implies the following
form for the unitary evolution operator
\begin{equation*}
  A=\int_{B}d\bk |\bk\>\<\bk|\otimes A_\bk,
\end{equation*}
where the the matrix
\begin{equation}\label{Ak}
A_\bk=\sum_{\bh\in S}e^{i\bh\cdot\bk}A_\bh
\end{equation}
is unitary for every $\bk$. Notice that $A_\bk$ is a matrix
polynomial in $e^{i\bh\cdot\bk}$. The unitarity conditions on $A_\bk$
for all $\bk\in B$ then read
\begin{align}
  &\sum_{\bh\in S}A_\bh A_\bh^\dag= \sum_{\bh\in S}A_\bh^\dag A_\bh=I_s,\label{eq:normalization}\\
  &\sum_{\bh-\bh'=\bh''}A_\bh A_{\bh'}^\dag   = \sum_{\bh-\bh'=\bh''}A_{\bh'}^\dag  A_\bh=0.
\label{eq:condunit}
\end{align}
The previous equations are a set of necessary and sufficient conditions for the unitarity of the time evolution, since they can be derived just imposing that the matrix $A_\bk$ is unitary. As explained in Sec.~\ref{sec:QW}, the requirement of isotropy for the QW needs the existence
of a group that acts transitively over the generator set $S_+$ with a faithful projective unitary representation that
satisfies Eq. (\ref{covW}). Notice that one has the identity
\begin{equation*}
\left(I\otimes A_{\bk=0}^\dag\right)A =\sum_{\bh\in S}T_\bh\otimes {A'}_\bh,
\end{equation*}
with $\sum_{\bh\in S}{A'}_\bh=I_s$, namely modulo a uniform local unitary we can always assume
\begin{equation}
  \sum_{\bh\in S}A_\bh=I_s,
  \label{eq:invvac}
\end{equation}
as explained in the following. Indeed, the isotropy requirement implies that $A_{\bk=0}$ commutes with the representation of the isotropy
group $L$, whence we can classify the QW by requiring identity (\ref{eq:invvac}) and then
multiplying the QW operator $A$ on the left by $(I\otimes V)$, with $V$ unitary commuting with the
representation of $L$. In the case that the representation is irreducible, then by Schur lemma we
have only $V=I_s$.

From now on we will restrict to $s=2$, which corresponds to the simplest
nontrivial QW in the case of $G$ Abelian. Indeed, in Ref.~\cite{PhysRevA.93.062334} it has been proved that if $G$ is an arbitrary Abelian group and $s=1$ (\emph{scalar} QW case), then the evolution is trivial.

\section{Imposing isotropy: admissible Cayley graphs of $\mathbb{Z}^d$}\label{sec:presentations}
In this Section we investigate how the isotropy assumption restricts the possible presentations of $G\cong \mathbb{Z}^d$. By Prop.~\ref{P1}, the isotropy groups are finite subgroups $L < \operatorname{Aut}(\mathbb{Z}^d)\cong \mathbb{GL}(d,\mathbb{Z})$: their action, by Cor.~\ref{cor2}, is defined to be transitive on the generating set $S_+$ and then is extended on all $\mathbb{Z}^d$ by linearity. Indeed,
%since once can always embed a Cayley graph of $\mathbb{Z}^d$ in $\mathbb{R}^d$, 
the generating set $S_+$ is the orbit of an arbitrary vector $\v{v}\in \mathbb{R}^d$ under the action of a finite subgroup $L<\mathbb{GL}(d,\mathbb{Z})$.

Let $M$ be a representation on integers of $L$ (so that $M_lM_f=M_{lf}$ for $l,f\in L$), and let us define the matrix $P\coloneqq \sum_{l\in L}M_l^{T}M_l$. For every $f\in L$  we have
\begin{equation}
\begin{aligned}\label{property}
PM_f &= \sum_{l\in L}M_{l}^{T}M_{lf} = \sum_{l'\in L} M_{l'f^{-1}}^{T}M_{l'} = \\
 &= \sum_{l'\in L} \left(M_{l'}M_{f^{-1}}\right)^{T}M_{l'} = M_{f^{-1}}^{T} P .
\end{aligned}
\end{equation}
Moreover, being a sum of positive operators, $P$ is also positive. Then, for $\ket{\eta} \in \operatorname{ker} P$, $\bra{\eta} P\ket{\eta} = \sum_{l\in L}\bra{\eta}M_l^{T}M_l\ket{\eta} = 0$ implies that $M_l \ket{\eta}=0$ $\forall l\in L$, namely $\ket{\eta}=0$ since all $M_l$ are invertible. Thus $P$ has trivial kernel and we can define the invertible change of representation:
\begin{align}\label{eq:orthrep}
\tilde{M_l} \coloneqq P^{1/2} M_l P^{-1/2}.
\end{align}
Using the definition of $P$ and property \eqref{property}, we obtain
\begin{align*}
\tilde{M_l}^{T}\tilde{M_l} &= P^{-1/2} M_l^{T} P M_l P^{-1/2} = \\
&= P^{-1/2} M_{l}^{T} M_{l^{-1}}^{T} P P^{-1/2} = I.
\end{align*}
This means that, as long as one embeds the Cayley graphs in $\mathbb{R}^d$, $L$ can always be represented orthogonally. Notice that the representation $\tilde{M}$ is in general on reals, namely $\lbrace \tilde{M}_l \rbrace_{l\in L} \subset \mathbb{O}(d,\mathbb{R})$ (from now on we denote it just as $\mathbb{O}(d)$). 

As one can find in Refs. \cite{mackiw1996finite,tahara_1971}, the finite subgroups of $\mathbb{GL}(d,\mathbb{Z})$ which are also subgroups of $\mathbb{O}(d)$ are isomorphic to:
\begin{itemize}
\item $d=3$: $\mathbb{Z}_n$, $D_{n}$ with $n\in \{1,2,3,4,6\}$, $A_4$, $S_4$, and the direct products of all the previous groups with $\mathbb{Z}_2$;
\item $d=2$: $\mathbb{Z}_n$ and $D_{n}$ with $n\in \{1,2,3,4,6\}$;
\item $d=1$: $\{e\}$ and $\mathbb{Z}_2$.
\end{itemize}
Accordingly, our cases of interest $d=1,2,3$ can be treated together, considering just $d=3$. We notice that for $d=1,2$ the finite subgroups of $\mathbb{GL}(d,\mathbb{Z})$ coincide with those of $\mathbb{O}(d)$, while for $d=3$ we restricted to those finite subgroups of $\mathbb{GL}(3,\mathbb{Z})$ that are also subgroups of $\mathbb{O}(3)$. 

A given generating set for $\mathbb{Z}^d$ satisfying the definition of isotropy can be constructed orbiting a vector in $\mathbb{R}^d$ under the aforementioned finite subgroups in $\mathbb{O}(d)$. Accordingly, given a presentation for $\mathbb{Z}^d$, if the associated Cayley graph satisfies isotropy then one can represent the generators having all the same Euclidean norm, namely they lie on a sphere centered at the origin: they form the orbit---which we will denote as $\mathcal{O}_L(\v{v})$---of an arbitrary $d$-dimensional real vector $\v{v}$ under the action of a finite subgroup $L<\mathbb{GL}(d,\mathbb{Z})$ represented in $\mathbb{O}(d)$.

In Appendix~\ref{app:excluding} we will consider the orbit of a vector $\v{v}\in \mathbb{R}^3$ under the real, orthogonal and three-dimensional faithful representations of $L$. Indeed, if we took into account also unfaithful representations, these would have nontrivial kernel---which is a normal subgroup---and the effective action on $\v{v}$ would be given by a faithful representation of the quotient group. Inspecting the subgroup structure of the finite subgroups of $\mathbb{GL}(3,\mathbb{Z})$, one can check that all the possible quotients are themselves finite subgroups of $\mathbb{GL}(3,\mathbb{Z})$~\footnote{This is straightforward as far as $\mathbb{Z}_n$ and $D_n$ are concerned; as for $A_4$ and $S_4$, one can verify it in a direct way considering their faithful representations given in Secs.~\ref{app:A_4} and~\ref{app:S_4}.}. Thus, the case of unfaithful representations is already considered as long as we take into account the faithful ones.

\section{The QWs with minimal complexity: the Weyl quantum walks}\label{Sec:Derivation}
In the following $X=V|X|$ will denote the polar decomposition of the operator $X$, with 
$|X|\coloneqq \sqrt{X^\dag X}$ the modulus of $X$, and $V$ unitary. Thus we will write the transition matrix
as
\begin{equation}
A_\bh=V_\bh|A_\bh|.
\end{equation}
From Eq.~\eqref{eq:condunit} with $\bh''=2\bh$ it follows that $A_\bh A_{-\bh}^\dag=0$, namely, $|A_\bh||A_{-\bh}|=0$. By definition the transition matrices are nonnull, hence 
$|A_\bh|$ and $|A_{-\bh}|$ must have orthogonal supports, and for $s=2$ they must then be
rank-one. Thus they can be written as follows
\begin{equation}
A_\bh \eqqcolon\alpha_{\bh}V_\bh|\eta_\bh\>\<\eta_\bh|,\quad A_{-\bh} \eqqcolon \alpha_{-\bh} 
V_{-\bh}|\eta_{-\bh}\>\<\eta_{-\bh}|,
  \label{eq:rankone}
\end{equation}
where $\{ |\eta_{+\bh}\>, |\eta_{-\bh}\>\}$ is an orthonormal basis and $\alpha_\bh> 0$. By the isotropy requirement we
have that for all $\bh,\bh'$ $\alpha_{\pm\bh}=\alpha_{\pm\bh'}\eqqcolon \alpha_\pm$. Furthermore, it is easy to see that we can choose $V_\bh=V_{-\bh}$ for every $\bh$~\footnote{We follow the argument of Ref.~\cite{DEPT15}. The condition $A_{\bh}^\dagger A_{-\bh}=0$ implies that $V_{\bh}V_{-\bh}^\dagger$ is diagonal in the basis $\{ |\eta_{+\bh}\>, |\eta_{-\bh}\> \}$. Since the transition matrices are not full rank, their polar decomposition is not unique: $V_{\bh} ( |\eta_{+\bh} {\>}{\<} \eta_{+\bh}| + e^{i\theta_{\bh}}|\eta_{-\bh} {\>}{\<} \eta_{-\bh}| )$ gives the same polar decomposition as $V_{\bh}$ $\forall \bh \in S$. Accordingly, one can tune the phases $\theta_{\pm \bh}$ to choose $V_{\bh}V_{-\bh}^\dagger=I$ $\forall \bh\in S$}.

Denoting the elements of $S_\pm$ as $\pm\bh_i$, suppose that there exists a subgroup $K\leq L$ such that, for some $\bh_1\in S_+$, $\forall \bh_i,\bh_j\in\mathcal{O}_K(\bh_1)$ with $\bh_i\neq \bh_j$, and for $\v{h}_l, \v{h}_m \in \{ \v{0}, \mathcal{O}_L(\v{h}_1) \}$, one has
\begin{equation}\label{e21}
\begin{split}
\!\bh_i-\bh_j=\bh_l-\bh_m\ \Longleftrightarrow \ (\bh_i=\bh_l) \vee( \bh_i=-\bh_m).
\end{split}
\end{equation}
Then, a second set of equations from conditions~\eqref{eq:condunit}
is 
\begin{align}
  & A_{\bh_1}A^\dag_{\bh_j}+A_{-\bh_j}A^\dag_{-\bh_1}=0,\label{eq:diagb}\\
  & A^\dag_{\bh_1}A_{\bh_j}+A^\dag_{-\bh_j}A_{-\bh_1}=0.\label{eq:diaga}
 \end{align}
Multiplying Eq.~\eqref{eq:diagb} by $A^ \dag_{\bh_j}$ on the left or
by $A_{\bh_1}$ on the right, we obtain
\begin{equation*}
  A^\dag_{\bh_j}A_{\bh_1}A^\dag_{\bh_j}=A_{\bh_1}A^\dag_{\bh_j}A_{\bh_1}=0.
\end{equation*}
Using the isotropy requirement and posing $A_{\v{h}_j} = U_kA_{\v{h}_i}U_k^\dagger$, we have
\begin{equation*}
  U_kA^\dag_{\bh_1}U_k^\dag A_{\bh_1}U_kA^\dag_{\bh_1}U_k^\dag=A_{\bh_1}U_kA^\dag_{\bh_1}U_k^\dag A_{\bh_1}=0.
\end{equation*}
By exploiting Eq.~\eqref{eq:rankone} both the previous equations become 
\begin{equation*}
  \<\eta_{\bh_1}|V^\dag_{\bh_1}U^\dag_k V_{\bh_1}|\eta_{\bh_1}\>\<\eta_{\bh_1}|U_k|\eta_{\bh_1}\>=0.
\end{equation*}
Then, at least one of the two following conditions must be satisfied
\begin{align}
  &\<\eta_{\bh_1}|V^\dag_{\bh_1}U^\dag_k V_{\bh_1}|\eta_{\bh_1}\>=0,\label{eq:adag}\\
  &\<\eta_{\bh_1}|U_k|\eta_{\bh_1}\>=0.\label{eq:a}
\end{align}
Furthermore, we remind that the representation $U$ can be chosen with unit determinant,
and for $s=2$ one has $U_k=\cos\theta I+i\sin\theta\,\bn_k\cdot\boldsymbol\sigma$.  Then, from Eqs.
(\ref{eq:adag}) and (\ref{eq:a}) one has $U_k=i\bn_k\cdot\boldsymbol\sigma$. Using the identity 
\begin{equation}
U_kU_{k'}=-\bn_k\cdot\bn_{k'}I-i(\bn_k\times\bn_{k'})\cdot\boldsymbol\sigma,
\end{equation}
it follows that all the $\bn_k$ must be mutually orthogonal and then $|K|\leq 4$. The case $K\cong \mathbb{Z}_3$ is not consistent with Eqs. \eqref{eq:adag} and \eqref{eq:a}. Accordingly, we end up with $K\in \{ I,\mathbb{Z}_2,\mathbb{Z}_2\times\mathbb{Z}_2,\mathbb{Z}_4 \}$. Notice that, up to a change of basis, one can always choose $\ket{\eta_{\pm \v{h}_1}}$ to be the eigenstates of $\sigma_Z$ without loss of generality. Then, by Eqs. \eqref{eq:adag},\eqref{eq:a} and imposing $U_k\in \mathbb{SU}(2)$ $\forall k\in K$, up to a change of basis it must be: either i) $U_K \coloneqq \operatorname{Rng}_K(U) = H$, where $H\coloneqq\{I,i\sigma_X,i\sigma_Y,i\sigma_Z\}$ is the Heisenberg group, or ii) $U_K =J$ where $J\in\{J_i\}_{i=1}^4$, where $J_1 \coloneqq \{ I,i\sigma_X \}$, $J_2 \coloneqq \{I,-V_{\bh_1}(i\sigma_X)V_{\bh_1}^\dagger \}$, $J_3\coloneqq \{ I,i\sigma_X,-I,-i\sigma_X \}$, and $J_4\coloneqq \{ I,-V_{\bh_1}(i\sigma_X)V_{\bh_1}^\dagger,-I,V_{\bh_1}(i\sigma_X)V_{\bh_1}^\dagger\}$, 
 or finally iii) $U_K=\{I\}$. We remark that $H$ is a projective faithful representation of $\mathbb{Z}_2\times \mathbb{Z}_2$ in $\mathbb{SU}(2)$, while $\{J_i\}_{i=1}^2$ are projective faithful representations of  $\mathbb{Z}_2$, while
 $\{J_i\}_{i=3}^4$ are unitary faithful representations of  $\mathbb{Z}_4$ in $\mathbb{SU}(2)$. We have thus proved the following result.
\begin{proposition}\label{prop1}
If the isotropy group $L$ contains a subgroup $K$ such that all the $\bh_k\in\mathcal{O}_K(\bh_1)$ (for $\bh_1\in S_+$) satisfy condition (\ref{e21}), then either $U_K=H$ or $U_K=J$ or $U_K=I$. 
\end{proposition}

\subsection*{The isotropic QWs on $\mathbb{Z}^d$ for $d=1,2,3$}\label{sec:remainingCayley}
In Appendix~\ref{app:excluding} we make use of Prop.~\ref{prop1} along with the unitarity constraints to exclude an infinite set of Cayley graphs arising from the aforementioned finite subgroups of $\mathbb{O}(3)$. We then proved the following.
\begin{proposition}
The primitive cells associated to the unique graphs admitting isotropic QWs in dimensions $d=1,2,3$ are those shown in Fig.~\ref{fig:Lattices}. 
\end{proposition}
Throughout the present section, we solve the unitarity conditions in dimension $d=1,2,3$ for the Cayley graphs associated to the primitive cells shown in Fig. $\ref{fig:Lattices}$, and for all the possible isotropy groups. We remind that in general each isotropy group gives rise to a distinct presentation for $\mathbb{Z}^d$, possibly with the same first-neighbours structure. As discussed in Fig.~\ref{fig:Lattices}, different presentations can be in general associated to the same primitive cell (one can include in $S_+$ the inverses or not). We will now prove our main result, which is stated in Prop. \ref{mainresult} after the following derivation.

\begin{figure*}[t]
    {\includegraphics[width=.22\linewidth]{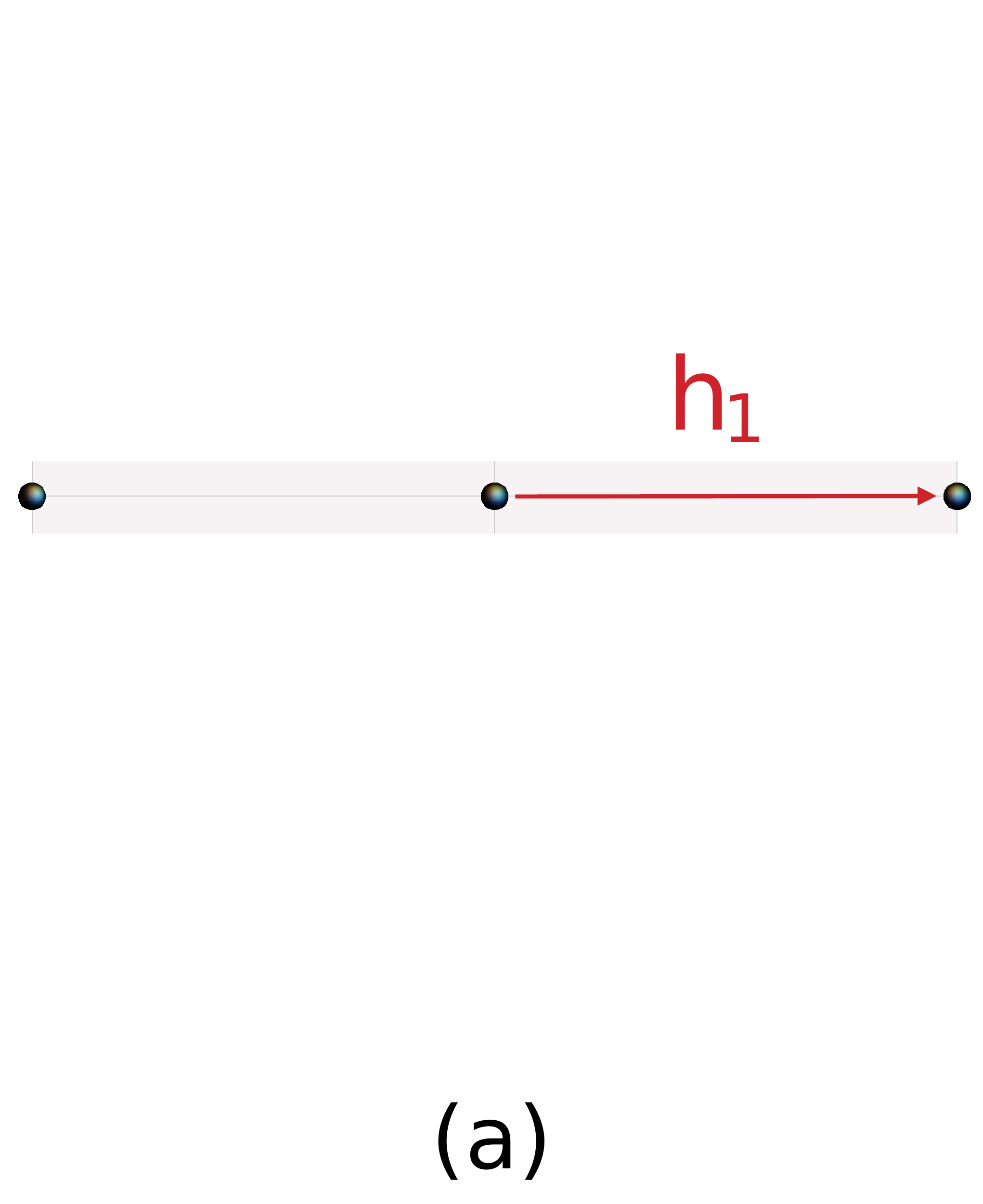}\phantomsubcaption}\qquad\qquad\qquad
    {\includegraphics[width=.22\linewidth]{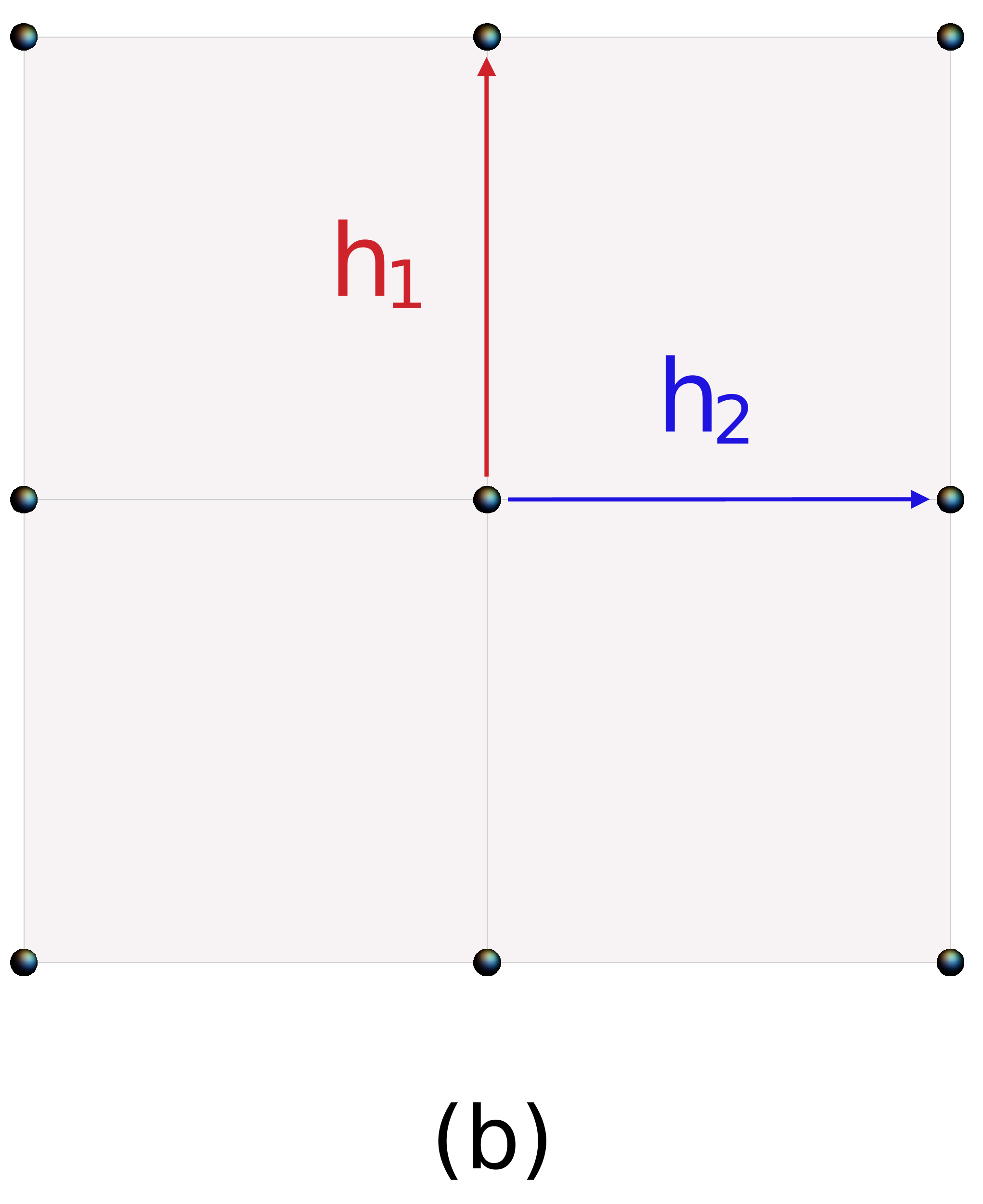}\phantomsubcaption}\qquad\qquad\qquad
    {\includegraphics[width=.22\linewidth]{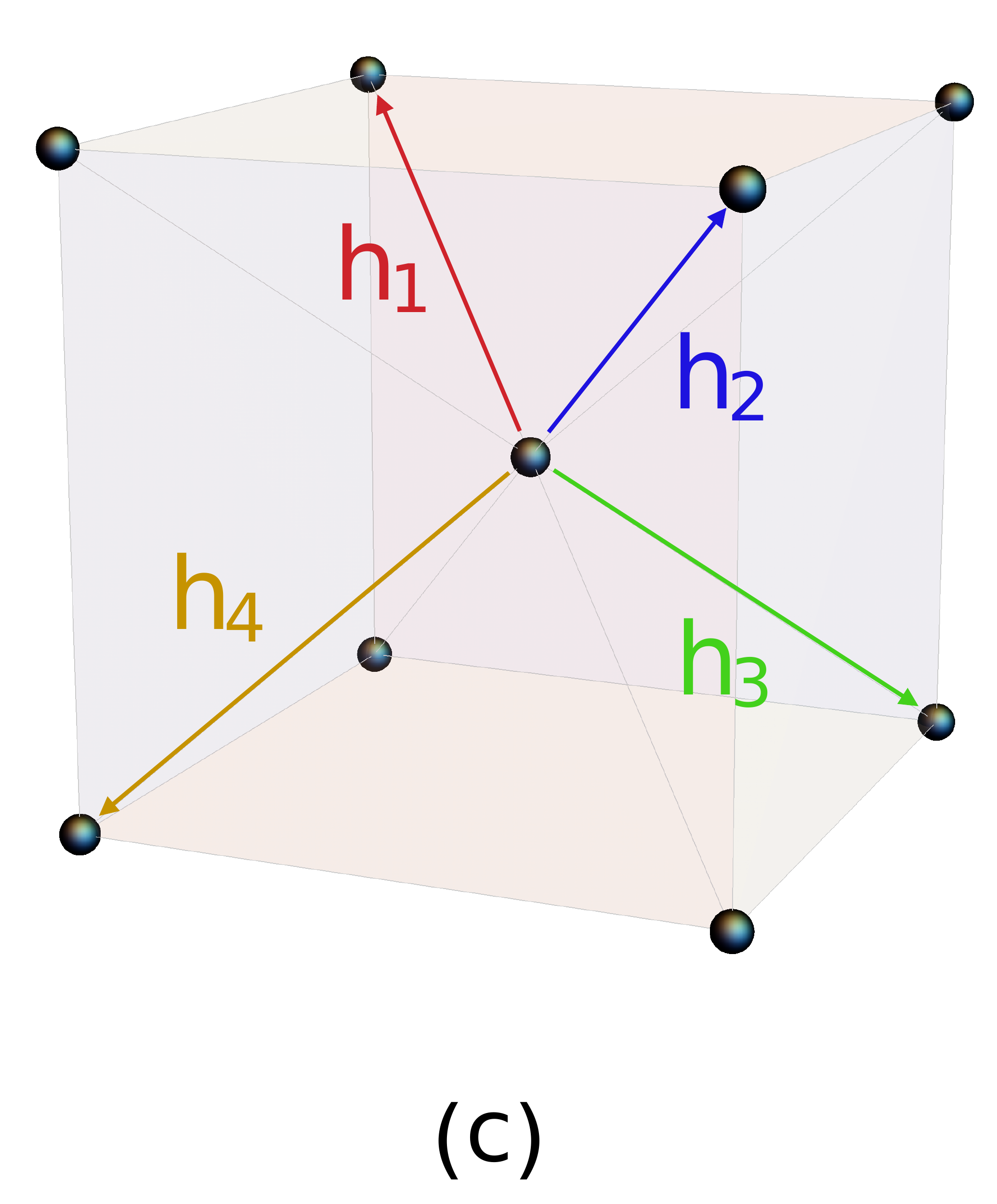}\phantomsubcaption}
 \caption{We report here the primitive cells of the unique graphs admitting isotropic QWs in dimensions $d=1,2,3$. Integer lattice (a): the isotropy groups can be $U_L= \{I\}$ and $U_L= \{I, i \sigma_X \}$, corresponding respectively to $S_+ = \{ \bh_1\}$ and $S_+ \equiv S_-= \{ \bh_1, -\bh_1\} $. Simple square lattice (b): the isotropy groups can be $U_L = \{ I, i\sigma_X \},\{ I, i\sigma_Z \}$ and $U_L = \{ I,i\sigma_X,i\sigma_Y, i\sigma_Z \}$, corresponding respectively to $S_+ = \{ \bh_1,\bh_2\}$ and $S_+ \equiv S_- = \{ \bh_1, \bh_2,-\bh_1, -\bh_2\} $. Body-centered cubic (BCC) lattice (c): the only possible isotropy group is $U_L = \{ I,i\sigma_X,i\sigma_Y, i\sigma_Z\}$, corresponding to $S_+ = \{ \bh_1,\bh_2,\bh_3,\bh_4\}$ with the nontrivial relator $\bh_1+\bh_2+\bh_3+\bh_4=0$. We notice that the case $d=1$ is the only one supporting the self-interaction, namely such that $A_e\neq 0$.}
  \label{fig:Lattices}
\end{figure*}

Before starting the derivation, we remind that in each case we can choose $|\eta_{\pm \bh_1}\>$ to be the eigenstates of $\sigma_Z$. Moreover, we will make use of Eq. \eqref{eq:rankone} to represent the transition matrices, reminding that $V_{\bh}=V_{-\bh}$. Finally, we recall that in Sec. \ref{sec:Zd} we showed that one can always impose condition \eqref{eq:invvac} and then multiply the transition matrices on the left by an arbitrary unitary commuting with the elements of the representation $U_L$.

\textbf{Case $d=1$.} We can write the transition matrices associated to $\pm \bh_1$ as
\begin{equation*}
A_{\bh_1} = \alpha_{+} V|\eta_{\bh_1}\> \< \eta_{\bh_1}|,\quad A_{-\bh_1} = \alpha_{-} V|\eta_{-\bh_1}\> \< \eta_{-\bh_1}|.
\end{equation*}
Multiplying on the right respectively by $A_{\v{h}_1}$ and $A_{-\v{h}_1}^\dagger$ the unitarity conditions
\begin{equation}\label{eq:unit1d}
\begin{split}
A_{\bh_1}A_e^\dagger + A_eA_{-\bh_1}^\dagger = 0, \\
A_e^\dagger A_{\bh_1} + A_{-\bh_1}^\dagger A_e = 0,
\end{split}
\end{equation}
one obtains
\begin{equation*}
A_{\pm \bh_1} A_e^\dagger  A_{\pm \bh_1} = 0,
\end{equation*}
which implies $A_e = V W$, where $W$ has vanishing diagonal elements in the basis $\{|\eta_{+\bh_1}\>, |\eta_{-\bh_1}\>\}$. Substituting into Eqs. \eqref{eq:unit1d}, one derives $\alpha_+ = \alpha_- \eqqcolon n$ and, up to a change of basis, $A_e = imV\sigma_X$ with $m\geq 0$. Imposing the normalization condition \eqref{eq:normalization} amounts to the relation $n^2 + m^2 = 1$. The admissible isotropy groups are $I$ and, up to a change of basis, $J_1$. Then, for $U_L = \{ I\}$, the transition matrices are given by:
\begin{equation*}
\begin{split}\label{transmat1d2}
&A_{\bh_1} = V \begin{pmatrix}
n & 0 \\ 0 & 0
\end{pmatrix},\quad A_{-\bh_1} = V \begin{pmatrix}
0 & 0 \\ 0 & n
\end{pmatrix}, \\
&A_e = V \begin{pmatrix}
0 & im \\ im & 0
\end{pmatrix},
\end{split}
\end{equation*}
where $V$ is an arbitrary unitary. For $U_L= \{I, i \sigma_X \}$, we impose condition \eqref{eq:invvac} and then $V$ can be taken as an arbitrary unitary commuting with $\sigma_X$.

\textbf{Case $d=2$.} The form of the transition matrices is:
\begin{equation*}
\begin{split}
& A_{\pm \bh_1} = \alpha_{\pm} V_{\bh_1}|\eta_{\pm \bh_1}\> \< \eta_{\pm \bh_1}|, \\
& A_{\pm \bh_2} = \alpha_{\pm} V_{\bh_2}|\eta_{\pm \bh_2}\> \< \eta_{\pm \bh_2}|.
\end{split}
\end{equation*}
Multiplying on the right by $A_{\v{h}_1}$ the unitarity conditions
\begin{equation}\label{eq:unit2d}
A_{\bh_1}A_{\pm \bh_2}^\dagger + A_{\mp\bh_2}A_{-\bh_1}^\dagger = 0,
\end{equation}
one obtains
\begin{equation*}
A_{\bh_1} A_{\pm \bh_2}^\dagger  A_{\bh_1} = 0.
\end{equation*}
The latter implies either i) $|\eta_{\pm \bh_1}\> = |\eta_{\pm \bh_2}\>$ or ii) $|\eta_{\pm \bh_1}\> = |\eta_{\mp \bh_2}\>$ and that, in both cases, one can choose $V_{\bh_1} = V_{\bh_2}(i\sigma_Y)$ up to a change of basis. In either cases, substituting into Eqs.~\eqref{eq:unit2d} one derives $\alpha_+ = \alpha_- \eqqcolon \alpha$ and, from the normalization condition \eqref{eq:normalization}, $\alpha = \frac{1}{\sqrt{2}}$. Redefining $V \coloneqq V_{\bh_2}$, in case i) one obtains the following family of transition matrices:
\begin{equation}\label{eq:transmat2d}
\begin{split}
A_{\pm \bh_1} = \pm \alpha V|\eta_{\mp \bh_1}\> \< \eta_{\pm \bh_1}|, \\
A_{\pm \bh_2} = \alpha V|\eta_{\pm \bh_1}\> \<  \eta_{\pm\bh_1}|.
\end{split}
\end{equation}
The second family, namely case ii), is connected to the first one via the exchange $\bh_2 \leftrightarrow -\bh_2$. One can check that the self-interaction term $T_e \otimes A_e$ is not supported by the unitarity conditions
\begin{equation*}
	A_{\bh} A_e^\dagger + A_e A_{-\bh}^\dagger = A_{\bh}^\dagger A_e + A_e^\dagger A_{-\bh} =0\quad \forall \bh\in S,
\end{equation*}
namely $A_e=0$. Imposing Eq. \eqref{eq:invvac}, one can choose
\begin{equation*}
V=\frac{1}{\sqrt{2}} \begin{pmatrix}
1 & 1 \\ -1 & 1
\end{pmatrix}
\end{equation*}
and then multiply the transition matrices by a unitary commuting with the representation $U_L$. The isotropy group can be either $J_2 \equiv \{ I, i\sigma_Z \}$ or $H$ for the first family of walks, while either $J_1 = \{ I, i\sigma_X \}$ or $H$ for the second one. Thus the first family is given by
\begin{equation*}
\begin{split}
A_{\bh_1} & = \frac{1}{2} V\begin{pmatrix}
1 & 0 \\
1 & 0 
\end{pmatrix},\quad A_{-\bh_1} = \frac{1}{2} V\begin{pmatrix}
0 & -1 \\
0 & 1 
\end{pmatrix}, \\
A_{\bh_2} & = \frac{1}{2} V\begin{pmatrix}
1 & 0 \\
-1 & 0 
\end{pmatrix},\quad A_{-\bh_2} = \frac{1}{2} V\begin{pmatrix}
0 & 1 \\
0 & 1 
\end{pmatrix},
\end{split}
\end{equation*}
where $V$ is either an arbitrary unitary commuting with $\sigma_Z$ or $V=I$, while the second family of transition matrices is obtained exchanging $\bh_2 \leftrightarrow -\bh_2$ and taking $V$ as either an arbitrary unitary commuting with $\sigma_X$ or $V=I$.

\textbf{Case $d=3$.} The isotropy requirement can be fulfilled with $U_L=H$. At least one of the two conditions of
Eqs.~\eqref{eq:adag} or \eqref{eq:a} must be fulfilled for any nontrivial $l\in L$. Since
Eq. (\ref{eq:a}) cannot be satisfied for $U_l = i\sigma_Z$, then it must be 
$\<\eta_{\bh_1}|V_{\bh_1}^\dag\sigma_ZV_{\bh_1}|\eta_{\bh_1}\>=0$. This implies
\begin{equation}
 \Tr[V_{\bh_1}^\dag\sigma_ZV_{\bh_1}\sigma_Z]=0.
 \label{eq:trnu}
\end{equation}
Writing $V_{\bh_1}$ in the general unitary form 
\begin{equation*}
  V_{\bh_1}=\theta
  \begin{pmatrix}
    \mu&-\nu^*\\
    \nu&\mu^*
  \end{pmatrix},
\end{equation*}
where $|\theta|^2=|\mu|^2+|\nu|^2=1$, the condition in Eq. \eqref{eq:trnu}
implies $|\mu|=|\nu|=2^{-1/2}$, and using the polar decomposition \eqref{eq:rankone} of
$A_{\pm\bh_1}$ we obtain
\begin{equation}\label{eq:transmatr3d}
  A_{\bh_1}=\frac {\alpha_+}{\sqrt2}
  \begin{pmatrix}
    \phi&0\\
    \psi&0
  \end{pmatrix},\quad 
  A_{-\bh_1}=\frac {\alpha_-}{\sqrt2}
  \begin{pmatrix}
    0&-\psi^*\\
    0&\phi^*
  \end{pmatrix},
\end{equation}
with $\phi,\psi$ phase factors. Using isotropy, namely considering the orbit of the above matrices under conjugation with $H$, we obtain
\begin{align}\label{eq:transmatr3d2}
\begin{split}
  &A_{\bh_2}=\frac {\alpha_+}{\sqrt2}
  \begin{pmatrix}
    0&\psi\\
    0&\phi
  \end{pmatrix},\quad
  A_{-\bh_2}=\frac {\alpha_-}{\sqrt2}
  \begin{pmatrix}
    \phi^*&0\\
    -\psi^*&0
  \end{pmatrix},\\
  &A_{\bh_3}=\frac {\alpha_+}{\sqrt2}
  \begin{pmatrix}
    0&-\psi\\
    0&\phi
  \end{pmatrix},\quad
  A_{-\bh_3}=\frac {\alpha_-}{\sqrt2}
  \begin{pmatrix}
    \phi^*&0\\
    \psi^*&0
  \end{pmatrix},\\
  &A_{\bh_4}=\frac {\alpha_+}{\sqrt2}
  \begin{pmatrix}
    \phi&0\\
    -\psi&0
  \end{pmatrix},\quad
  A_{-\bh_4}=\frac {\alpha_-}{\sqrt2}
  \begin{pmatrix}
    0&\psi^*\\
    0&\phi^*
  \end{pmatrix}.
\end{split}
\end{align}
Also in this case, the self-interaction term is not supported by the unitarity conditions. Finally, we can write the matrix $A_\bk$ in Eq. \eqref{Ak} as
\begin{equation*}
A_\bk =\sum_{i=1}^4 (A_{\bh_i}e^{ik_i}+A_{-\bh_i}e^{-ik_i})
\end{equation*}
and imposing unitarity of $A_\bk$ for every $\bk$, one obtains the following conditions
\begin{equation*}
  \alpha_+^2=\alpha_-^2=\frac14,\ \phi^{*2}+\phi^2=\psi^{*2}+\psi^2=0,
\end{equation*}
namely 
\begin{equation*}
\phi,\psi\in\left\{\pm\zeta^+\coloneqq \pm\frac{1+ i}{\sqrt{2}}, \pm\zeta^-\coloneqq \pm\frac{1- i}{\sqrt{2}}\right\}.
\end{equation*}
The different choices of the overall signs for $\phi,\psi$ are connected to each other by an overall phase factor and by unitary conjugation by $\sigma_Z$. Then we can fix then choosing the plus signs. The choices $\phi=\zeta^\pm,\psi=\zeta^\mp$ are equivalent to $\phi=\psi=\zeta^\pm$ via conjugation of the former by $e^{\pm i\tfrac\pi4\sigma_Z}$ and an exchange $\bh_1 \leftrightarrow \bh_4$. Accordingly, the QWs found are given by the transition matrices of Eqs. \eqref{eq:transmatr3d} and \eqref{eq:transmatr3d2} with $\psi = \varphi = \zeta^{\pm}$, namely the two Weyl QWs presented in Ref.~\cite{DP14}.

We have thus proved the following main result.
\begin{proposition}[Classification of the isotropic QWs on lattices of dimension $d=1,2,3$ with a coin system of dimension $s=2$]\label{mainresult}
Let $S=S_+\cup S_-\cup \{e\}$ denote a set of generators for $\mathbb{Z}^d$ and let $\{ A_{\bh} \}_{\bh\in S}$ denote the set of transition matrices of a QW on $\mathbb{Z}^d$ with a coin system of dimension $s=2$ and isotropic on $S_+$. Then for each $d=1,2,3$ the admissible graphs are unique (see Fig. \ref{fig:Lattices}) and one has the following:
\begin{itemize}
\item[a)] Case $d=1$:
\begin{align*}
\begin{split}
&A_{\bh_1} = V \begin{pmatrix}
n & 0 \\ 0 & 0
\end{pmatrix},\quad A_{-\bh_1} = V \begin{pmatrix}
0 & 0 \\ 0 & n
\end{pmatrix}, \\
&A_e = V \begin{pmatrix}
0 & im \\ im & 0
\end{pmatrix},
\end{split}
\end{align*}
where $n,m$ are real such that $n^2+m^2=1$, and $V$ is an arbitrary unitary if $S_+= \{\bh_1\}$ or $V$ is a unitary commuting with $\sigma_X$ if $S_+= \{\bh_1,-\bh_1\}$.
\item[b)] Case $d=2$: one has $A_e=0$ and
\begin{align*}
\begin{split}
&A_{\bh_1}  = \frac{1}{2} V\begin{pmatrix}
1 & 0 \\
1 & 0 
\end{pmatrix},\quad A_{-\bh_1} = \frac{1}{2} V\begin{pmatrix}
0 & -1 \\
0 & 1 
\end{pmatrix}, \\
&A_{\bh_2}  = \frac{1}{2} V\begin{pmatrix}
0 & 1 \\
0 & 1 
\end{pmatrix},\quad A_{-\bh_2} = \frac{1}{2} V\begin{pmatrix}
1 & 0 \\
-1 & 0 
\end{pmatrix},
\end{split}
\end{align*}
where $V$ is a unitary commuting with $\sigma_X$ if $S_+= \{\bh_1,\bh_2\}$ or $V=I$ if $S_+= \{\bh_1,\bh_2,-\bh_1,-\bh_2\}$.
\item[c)] Case $d=3$: one has $A_e=0$ and
\begin{align*}
\begin{split}
  &A_{\bh_1}= 
  \begin{pmatrix}
    \eta^{\pm}&0\\
    \eta^{\pm}&0
  \end{pmatrix},\quad
  A_{-\bh_1}= 
  \begin{pmatrix}
    0&-\eta^{\mp}\\
    0&\eta^{\mp}
  \end{pmatrix},\\
  &A_{\bh_2}= 
  \begin{pmatrix}
    0&\eta^{\pm}\\
    0&\eta^{\pm}
  \end{pmatrix},\quad
  A_{-\bh_2}= 
  \begin{pmatrix}
    \eta^{\mp}&0\\
    -\eta^{\mp}&0
  \end{pmatrix},\\
  &A_{\bh_3}= 
  \begin{pmatrix}
    0&-\eta^{\pm}\\
    0&\eta^{\pm}
  \end{pmatrix},\quad
  A_{-\bh_3}= 
  \begin{pmatrix}
    \eta^{\mp}&0\\
    \eta^{\mp}&0
  \end{pmatrix},\\
  &A_{\bh_4}= 
  \begin{pmatrix}
    \eta^{\pm}&0\\
    -\eta^{\pm}&0
  \end{pmatrix},\quad
  A_{-\bh_4}= 
  \begin{pmatrix}
    0&\eta^{\mp}\\
    0&\eta^{\mp}
  \end{pmatrix}.
\end{split}
\end{align*}
where $\eta^{\pm} = \frac{1\pm i}{4}$ and $S_+ = \{\bh_1,\bh_2,\bh_3,\bh_4\}$ with the nontrivial relator $\bh_1+\bh_2+\bh_3+\bh_4=0$.
\end{itemize}
\end{proposition}

\section{Conclusions}\label{Sec:conclu}
In this paper we presented a complete classification of the isotropic quantum walks on lattices of dimension $d=1,2,3$ with coin dimension $s=2$. We have extended the isotropy definition of Ref.~\cite{DP14}, to account for groups with generators of different orders. We introduced a technique to construct the Cayley graphs of a given group $G$ satisfying a relevant necessary condition for isotropy. This allowed us to exclude an infinite class of Cayley graphs of $\mathbb{Z}^d$. The technique is sufficiently flexible to be used in the future for other generally non Abelian groups. Remarkably, the Cayley graph is unique for each dimension $d=1,2,3$ and for $d=3$ the only admissible QWs are the two Weyl QWs presented in Ref.~\cite{DP14}.
The use of isotropy since the very beginning has made the solution of the unitarity equations significantly shorter. Moreover, we eliminated the superfluous technical assumption used in Ref.~\cite{DP14} mentioned in the Introduction. In consideration of the length of the derivation from informational principles of the Weyl equation in Ref.~\cite{DP14}, the present derivation constitutes a thoroughly independent check. Finally, this result represents the extension of the classification of Ref.~\cite{PhysRevA.93.062334}.
\section*{Acknoledgments}
This publication was made possible through the support of a grant from the John Templeton Foundation, ID \# 60609 ``Quantum Causal Structures''. The opinions expressed in this publication are those of the authors and do not necessarily reflect the views of the John Templeton Foundation.

\appendix

\section{Excluding Cayley graphs}\label{app:excluding}
In Secs. \ref{app:AS_4}--\ref{app:Z_2} we will exclude the infinite family of graphs arising from the following finite isotropy groups $L< \mathbb{O}(3)$:
\begin{enumerate}
\item\label{i1} $A_4$, $S_4$ and their direct product with $\mathbb{Z}_2$ (except for the cases in item \ref{i2});
\item\label{i2} the special instances of item \ref{i1} where the orbits contain the vertices of a truncated tetrahedron;
\item$\mathbb{Z}_n,D_{n}$ for $n=3,4,6$ and their direct product with $\mathbb{Z}_2$;
\item one special instance arising from $D_2$, $D_2 \times \mathbb{Z}_2$.
\end{enumerate}

\subsection{Excluding $A_4$- and $S_4$-symmetric Cayley graphs}\label{app:AS_4}
In this subsection we use the convention that unwritten matrix elements are zero. In Secs.~\ref{app:A_4} and \ref{app:S_4} we will consider the orbit of an arbitrary three-dimensional vector $\v{v}= (\alpha, \beta, \gamma)^T $ under the action of the finite groups $L\cong A_4,S_4$ in $\mathbb{O}(3)$. To this purpose, as discussed in Sec.~\ref{sec:presentations}, we will use the real, orthogonal and three-dimensional faithful representations of $L$, identifying its representation with the group itself. In the present case of $L \cong A_4,S_4$, such representations coincide with the irreducible ones, since the reducible ones cannot be faithful (otherwise they would have orthogonal blocks of dimension at most 2, but $A_4,S_4$ are not subgroups of $\mathbb{O}(2)$).

We denote with $\mathcal{O}_{L}(\v{v})$  the family of orbits of $\v{v}$ under the action of $L$, parametrized by $\alpha,\beta,\gamma$. Each orbit satisfies a necessary condition to give rise to an isotropic presentation for $\mathbb{Z}^d$ for $d=1,2,3$.

\begin{proposition}\label{P3}
If $L$ contains a ternary subgroup $K\cong\mathbb{Z}_3$ such that for $\bh_i,\bh_j\in \mathcal{O}_{K}(\v{v})$ and 
$\bh_l,\bh_m\in \mathcal{O}_{L}(\v{v})$,  the condition in Eq. (\ref{e21}) is satisfied,  then the set of vertices $\mathcal{O}_{L}(\v{v})$ cannot satisfy  the necessary conditions \eqref{eq:diaga},\eqref{eq:diagb} for unitarity. 
\end{proposition}
\Proof  By Prop.~\ref{prop1}, $K$ has to be a subgroup of the Heisenberg group $H$. However $H$ does not contain ternary subgroups. \qed

We will make use of Prop.~\ref{P3} to exclude an infinite family of presentations arising from $L\cong A_4,S_4$. Since by Eq. \eqref{e21} we are interested in sums or differences of generators, the cases $L \cong A_4\times \mathbb{Z}_2,S_4\times \mathbb{Z}_2$ are already accounted: their irreducible representations just add the inversion to the irreducible ones of $A_4,S_4$.

The groups $L$ contain four isomorphic copies of $\mathbb{Z}_3$ (see Subsecs. \ref{app:A_4},\ref{app:S_4}). Let us denote with $D$ the generator of one of this cyclic subgroups. The content of Eqs. \eqref{e21} for a fixed choice of $i,j$ translates to the following. Suppose that for all $A,B\in L_0:=\{ 0\in M_3(\mathbb{R})\}\cup L$ one has:
\begin{equation}\label{difference}
(I-D)\v{v} = s(A+tB)\v{v}\Leftrightarrow (sA\v{v}=\v{v})\vee(stB\v{v}=\v{v}), 
\end{equation}
 ($s,t$ signs). Our strategy is now to solve the necessary conditions for the violation of \eqref{difference}, consisting in systems of the form
\begin{equation}\label{system}
\forall A,B\in L_0,\;(I-D-s(A+tB))\v{v}=0.
\end{equation}
These will produce some solutions $\v{v}_0$. Then we can choose another vector in $\mathcal{O}_{L}(\v{v}_0)$, impose again Eq. \eqref{system}, and iterate until we end up either with the trivial solution, or with a  system of linear equations for $\alpha, \beta, \gamma$. By Prop. \ref{P3}, the only $A_4$- or $S_4$-symmetric Cayley graphs of $\mathbb Z^3$ for which the unitarity conditions may be satisfied must then be found among the non-trivial solutions of the above systems. Since the condition \eqref{system} is only necessary, we need to check whether the solutions actually violate condition~\eqref{difference}. The remaining differences $(D-D^2)\v{v}$ and $(D^2-I)\v{v}$ are the orbit of $(I-D)\v{v}$ under $D$, then we can just solve~\eqref{system} and check~\eqref{difference}.

In the following we will show that \eqref{difference} has only trivial solutions for $A,B\in L$, except for the special case where $\v v=\alpha(3,1,1)^T$, that will be treated separately in Subsec.~\ref{app:ExclusionA4}. At the end of Subsec.~\ref{app:S_4} we will then prove the same result in the case of $B=0$.

It is useful to notice the following:
\begin{remark}\label{rem1}
$\v{v}_1\in \mathbb{R}^3$ solves
\begin{equation*}
(I-D-s(A+tB))\v{v}_1=0
\end{equation*}
iff $\v{v}_2 \coloneqq F_2^{-1}\v{v}_1$ solves
\begin{equation*}
F_1(I-D-s(A+tB))F_2\v{v}_2=0,
\end{equation*}
for some arbitrary $F_1,F_2\in\mathbb{GL}(3,\mathbb{R})$.
In particular, this is relevant in the case $F_2\in L$, because it means that the orbits generated by the two solutions 
$\v{v}_1, \v{v}_2$ coincide.\end{remark}
This remark will allow us to considerably reduce the number of systems we have to solve. In the following we will refer to a particular solution for \eqref{system} indifferently with: 1) the solution vector $\v{v}_0$, or 2) the lattice which $\v{v}_0$ gives rise to, or 3) the polyhedron whose vertices are the elements of $\mathcal{O}_{L}(\v{v}_0)$, or 4) any other vector in $\mathcal{O}_{L}(\v{v}_0)$, or finally 5) the orbit $\mathcal{O}_{L}(\v{v}_0)$. The cases we will end up with are the following:
\begin{enumerate}
	\item The simple cubic lattice, generated orbiting $\v{v}_{s}= \alpha (1,0,0)^T $ under $A_4$: its vertices are all the signed permutations of the coordinates of $\v{v}_s$.
	\item The BCC lattice, generated orbiting $\v{v}_b =\alpha (1,-1,-1)^T$ under $S_4$: its vertices are all the signed permutations of the coordinates of $\v{v}_b$.
	\item The cuboctahedron, whose vertices are all the signed permutations of the coordinates of $\v{v}_c = \alpha (1,-1,0)^T $ and are generated by orbiting $\v{v}_c$ under $A_4$.
	\item The truncated tetrahedron, whose vertices are all the permutations with an even number of minus signs of the coordinates of $\v{v}_{tt}= \alpha (3,1,1)^T $ and are generated by orbiting $\v{v}_{tt}$ under $A_4$; in addition, one can also find the solution including the inverses, which is given by $\mathcal{O}_{S_4}(\v{v}_{tt})$.
	\item The truncated octahedron, whose vertices are all the signed permutations of the coordinates of $\v{v}_{to} = \alpha (1,-2,0)^T $ and are generated by orbiting $\v{v}_{to}$ under $S_4$.
\end{enumerate}
One can easily check that $\mathcal{O}_L(\v{v}_0)$ for the five cases above actually are generating sets for some presentation of $\mathbb{Z}^3$.

In the following, we will choose $D=R$ with $R(x,y,z)^T=(z,x,y)^T$ ($R$ is contained in the representation of both $A_4$ and $S_4$). As a consequence, we can consider $A\neq B$, since otherwise there are two possible cases:
\begin{enumerate}
	\item $(I-R)\v{v} = \pm 2A\v{v}$, implying $(A^{-1}-A^{-1}R)\v{v} = \pm 2\v{v}$. Since $A,R\in\mathbb{O}(3)$, by the triangle inequality it must be
	\begin{equation*}
	A^{-1}\v{v} = \pm \v{v},\ A^{-1}R\v{v} = \mp \v{v},
	\end{equation*}
	and in particular $\v{v} = -R\v{v}$ holds. This implies $\v{v} = (0,0,0)^T $.
	\item $(I-R)\v{v} = 0$, implying $\v{v} = \alpha (1,1,1)^T $.
\end{enumerate}

Finally, the reader can check that for $\v{v}_0 \in\{\v{v}_{s}, \v{v}_{b}, \v{v}_{c}, \v{v}_{to}\}$ condition (\ref{difference}) is not violated,
thus excluding the cases of $S_+=\mathcal O_L(\v v_0)$ by virtue of Prop.~\ref{P3}.

\subsubsection{Excluding $A_4$-symmetric Cayley graphs}\label{app:A_4}
$A_4$ has a unique three-dimensional real irreducible representation, generated by the matrices:
\begin{equation}\label{A4rep}
X_1 = \begin{pmatrix}
1&0&0 \\ 0&-1&0 \\ 0&0&-1
\end{pmatrix},\ R = \begin{pmatrix}
0&0&1 \\ 1&0&0 \\ 0&1&0
\end{pmatrix}.
\end{equation}
We define
\begin{equation*}
X_0 = I,\ X_2 = RX_1R^{-1},\ X_3 = R^2X_1R^{-2}.
\end{equation*}
The group contains four isomorphic copies of $\mathbb{Z}_3$, generated respectively by the elements of the set $\{ R, X_1R, X_2R, X_3R \}$ (these are cyclic signed permutations of the coordinates).

We now choose the subgroup generated by $R$ and consider the difference $(I-R)\v{v}$, setting the condition \eqref{system} for any $A,B\in A_4$. Each of these define linear systems of three equations for $\v{v}$.
% We now derive the conditions on $\v{v}$ for the existence of a non-trivial pair $\v v_1,\v v_2\in\mathcal O_{A_4}(\v v)$ such that $(I-R)\v{v}=\v v_1-\v v_2$, namely such that there exists at least one pair $A,B\in {A_4}$ such that Eq. \eqref{system} is satisfied.
If $A$ equals $I$ or $R$, then it is easy to see that $\exists G \in A_4$ such that $G\v{v} = s \v{v}$ ($s$ a sign): this implies that either $\v{v}= (0,\beta,\gamma)^T $ up to signed permutations, or $\mathcal{O}_{A_4}(\v{v})=\mathcal{O}_{A_4}(\v{v}_b)$. The latter case was excluded in Subsec.~\ref{app:AS_4}. The remaining cases are then i) $A,B\not\in\{I,R\}$ or ii) $\v v=(0,\beta,\gamma)^T$ and signed permutations.
Case (ii), however, will appear as a special instance of (i).
In case (i), we have six cases for $s(A+tB)$:
\begin{enumerate}
 \item $s(X_i+tX_j) = \begin{pmatrix}
 2s&& \\ &\pm \xi& \\ &&0
 \end{pmatrix}$, modulo permutations of the diagonal elements, with arbitrary sign $s$ and for $\xi \coloneqq 0,2$.
 \item $s(X_i+tX_jR) = \begin{pmatrix}
 s_1&0&t_1 \\ t_2&s_2&0 \\ 0&t_3&s_3
 \end{pmatrix}$, with $s_1s_2+s_1s_3+s_2s_3=t_1t_2+t_1t_3+t_2t_3=-1$.
 \item $s(X_i+tX_jR^2) = \begin{pmatrix}
 s_1&t_1&0 \\ 0&s_2&t_2 \\ t_3&0&s_3
 \end{pmatrix}$, with arbitrary signs $t_k$, and $s_1s_2+s_1s_3+s_2s_3=-1$.
 \item $s(X_i+tX_j)R = \begin{pmatrix}
 && 2s \\ \pm \xi && \\ &0&
 \end{pmatrix}$ and permutations of the written elements, with arbitrary sign $s$ and $\xi = 0,2$.
 \item $s(X_i+tX_jR)R = \begin{pmatrix}
 0&t_1&s_1 \\ s_2&0&t_2 \\ t_3&s_3&0
 \end{pmatrix}$, with arbitrary signs $t_k$, and $s_1s_2+s_1s_3+s_2s_3=-1$.
 \item $s(X_i+tX_j)R^2 = \begin{pmatrix}
 & 2s& \\ &&\pm \xi \\ 0&&
 \end{pmatrix}$ and permutations of the written elements, with arbitrary sign $s$ and $\xi = 0,2$.
\end{enumerate}

All the above mentioned permutations of elements and those between the $s_i$ and $t_i$ are performed by conjugation with $R^{\pm 1}$. Since
\[
(I-R-sR(A+tB)R^{-1}) = R(I-R-s(A+tB))R^{-1},
\]
by Remark \ref{rem1} we can just choose one permutation in each of the six cases to find the orbits of the solutions.

Accordingly, explicitly computing the expression
\begin{equation*}
I-R-s(A+tB) = \begin{pmatrix}
1&0&-1 \\ -1&1&0 \\ 0&-1&1
\end{pmatrix} - s(A+tB),
\end{equation*}
we end up with the following cases:
\begin{enumerate}
	\item $\begin{pmatrix}
	1+2s&0&-1 \\ -1&1\pm \xi&0 \\ 0&-1&1
	\end{pmatrix}$, for $s$ arbitrary sign.
	\item $\begin{pmatrix}
	2&&-2 \\ -\xi '&\xi& \\ &0&0
	\end{pmatrix}, \begin{pmatrix}
	2&&0 \\ -2&\xi& \\ &-\xi '&0
	\end{pmatrix}, \begin{pmatrix}
	0&&-2 \\ -\xi '&2& \\ &0&\xi
	\end{pmatrix}$, \\ with $\xi ,\xi ' =0,2$.
	\item $\begin{pmatrix}
	2&s_1&-1 \\ -1&\xi&s_2 \\ s_3&-1&0
	\end{pmatrix}$, with $s_i$ arbitrary.
	\item $\begin{pmatrix}
	1&0&2s-1 \\ -1&1&0 \\ 0&\pm \xi -1&1
	\end{pmatrix}$, with $s$ arbitrary.
	\item $\begin{pmatrix}
	1&s_1&0 \\ -2 &1&s_2 \\ s_3&-\xi &1
	\end{pmatrix}$, with $s$ arbitrary.
	\item $\begin{pmatrix}
	1&2s&-1 \\ -1&1&0 \\ \pm \xi&-1&1
	\end{pmatrix}$, with $s$ arbitrary.
\end{enumerate}
The only solution to cases 1 and 4 is $\mathcal O_{A_4}(\v v_b)$. Cases 3, 5 and 6 can be treated together since they exhibit a common structure: their solutions are $\mathcal O_{A_4}(\v v_b)$ (which has been already excluded by Prop. \ref{P3}) and $\mathcal O_{A_4}(\v v_{tt})$ (which is excluded in Sec.~\ref{app:ExclusionA4}). The only relevant case is 2, since all the other cases have been already excluded.

In case 2, the most general orbits of solutions are $\mathcal{O}_{A_4}(\v{v}_i)$ for $i=1,2,3$, where
\begin{equation}\label{solutions}
\begin{split}
\v{v}_1 = \begin{pmatrix} \alpha \\ \beta \\ \alpha \end{pmatrix},\ \v{v}_2 = \begin{pmatrix} 0 \\ \beta \\ \gamma \end{pmatrix},\ \v{v}_3 = R^2\v{v}_2.
\end{split}
\end{equation}
Nevertheless, for $\v{v}\in\{\v{v}_2,\v{v}_3\}$ the condition (\ref{difference}) is not violated. Indeed,  $\v{v}_2$ was found as a solution of 
\begin{equation}
(I-R+X_1-RX_1)\v{v}_2=0,
\end{equation}
however $X_1\v{v}_2=-\v{v}_2$, and thus Eq. (\ref{difference}) is satisfied. A similar argument holds for $\v{v}_3$. By virtue of  Prop. \ref{P3} the corresponding orbits $\mathcal{O}_{A_4}(\v{v}_2)$ and $\mathcal{O}_{A_4}(\v{v}_3)$ are excluded. 

From the above analysis we already know that the only relevant solution is $\v{v}_1$ for case 2, modulo cyclic permutations. We now impose that $X_1\v{v}_1$, which is in $\mathcal{O}_{A_4}(\v{v}_1)$, is itself a solution of Eq.~\eqref{system}. Thus we impose
\[
X_1\v{v}_1=\v{w}\in\left\{ 
\begin{pmatrix}
\alpha'  \\ \beta' \\ \alpha'
\end{pmatrix}, \begin{pmatrix}
\alpha' \\ \alpha' \\ \beta' 
\end{pmatrix} ,\begin{pmatrix}
\beta' \\ \alpha' \\ \alpha'
\end{pmatrix}\right\}.
\]
The solutions are $\mathcal{O}_{A_4}(\v{v}_s)$ and $\mathcal{O}_{A_4}(\v{v}_b)$: we can exclude also this last case. 

\subsubsection{Excluding $S_4$-symmetric Cayley graphs}\label{app:S_4}
The group $S_4$ contains $A_4$ as a subgroup of index 2. The element connecting the two cosets is an involution, which we will denote with $C$. $S_4$ has two three-dimensional irreducible representatons: their elements are signed permutations matrices of three elements and the two representations coincide up to a minus sign on the elements in the coset $CA_4$. Nevertheless, in our case the sign is irrelevant, since we are considering combinations $s(A+tB)$ of $A,B\in S_4$ with $s,t$ arbitrary signs. Accordingly, we consider the representation resulting from orbiting the elements generated by \eqref{A4rep} under the left action of $\{ I, C \}$ with
\begin{equation*}
C = \begin{pmatrix}
0&0&1 \\ 0&1&0 \\ 1&0&0
\end{pmatrix},
\end{equation*}
whose effect is just an exchange of the first and third row.

Let us now define $X_i' \coloneqq CX_i$. In order to perform the computation of $s(A+tB)$, we proceed as follows. We have to compute
\begin{equation}\label{combinations}
\begin{split}
s(X_i+tX_j), \ s(X_i'+tX_j'),\\
s(X_i+tX_jR),\ s(X_i'+tX_j'R), \\
s(X_i'+tX_j),\ s(X_i'+tX_jR),\ s(X_i'+tX_jR^2) \\
\end{split}
\end{equation}
and then recover all the remaining combinations by right multiplication of these by $R^{\pm 1}$. For \eqref{combinations}, we obtain the following cases:
\begin{enumerate}
 \item $\begin{pmatrix}
 \pm 2&& \\ & \xi & \\ &&0
 \end{pmatrix}, \begin{pmatrix}
 && \pm 2 \\ & \xi& \\ 0&&
 \end{pmatrix},$ considering all the permutations of elements and $\xi = 0,\pm 2$.
 \item $\begin{pmatrix}
 s_1&&t_1 \\ t_2&s_2& \\ &t_3&s_3
 \end{pmatrix}, \begin{pmatrix}
 &t_1&s_1 \\ t_2&s_2& \\ s_3&&t_3
 \end{pmatrix},\begin{pmatrix}
 s_1&&s_2 \\ & \xi& \\ s_3&&s_4
 \end{pmatrix}, $ \\
 $\begin{pmatrix}
 && \xi \\ t_2&s_2& \\ s_3&t_3&
 \end{pmatrix},\begin{pmatrix}
&t_1&  s_1\\ &s_2&t_2 \\ \xi && 
 \end{pmatrix}, $ for $\xi =0, \pm 2$.
\end{enumerate}
As mentioned above, one has to add to these cases the matrices resulting from a right multiplication of the previous ones by $R^{\pm 1}$, whose action is a cyclic permutation of the columns. Let us now consider
\[
I-R = \begin{pmatrix}
 1&0&-1 \\ -1&1&0 \\ 0&-1&1
\end{pmatrix},
\]
and derive the following matrices
\begin{equation}\label{matrices}
I-R+s(A+tB)R^i,\quad i=0,\pm 1
\end{equation}
for all the mentioned cases.
\begin{enumerate}
	\item It's easy to verify that, in this case, either the matrices in Eq. \eqref{matrices} have trivial solution or their solutions are $\mathcal{O}_{S_4}(\v{v}_b)$ (already excluded) and $\mathcal{O}_{S_4}(\v{v}_{tt})$ (which will be treated in Sec.~\ref{app:ExclusionA4}).
	\item $\begin{pmatrix}
	1+s_1&&t_1-1 \\ t_2-1&1+s_2& \\ &t_3-1&1+s_3
	\end{pmatrix}, \begin{pmatrix}
	1&& \xi-1 \\ t_2-1&1+s_2& \\ s_3&t_3-1&1
	\end{pmatrix},$ \\ $\begin{pmatrix}
	1&t_1&s_1-1 \\ s_2-1&1&t_2 \\ t_3&s_3-1&1
	\end{pmatrix},\begin{pmatrix}
	1& \xi&-1 \\ s_2-1&1&t_2 \\ t_3&-1&1+s_3
	\end{pmatrix},$ \\
	$\begin{pmatrix}
	1+t_1&s_1&-1 \\ -1&1+t_2&s_2 \\ s_3&-1&1+t_3
	\end{pmatrix}, \begin{pmatrix}
	1+\xi&& -1 \\ -1&1+t_2&s_2 \\ &s_3-1&1+t_3
	\end{pmatrix}$, \\
	$\begin{pmatrix}
	1&t_1&s_1-1 \\ t_2-1&1+s_2& \\ s_3&-1&1+t_3
	\end{pmatrix}, \begin{pmatrix}
	1+t_1&s_1&-1 \\ s_2-1&1+t_2& \\ &-1& \xi +1
	\end{pmatrix}$ \\
	$\begin{pmatrix}
	1+t_1&s_1&-1 \\ s_2-1&1&t_2 \\ &t_3-1&1+s_3
	\end{pmatrix}, \begin{pmatrix}
	1+s_1&&t_1-1 \\ t_2-1&1&s_2 \\ & \xi -1&1
	\end{pmatrix}$ \\
	$\begin{pmatrix}
	1+s_1&&t_1-1 \\ -1&1+t_2&s_2 \\ t_3&s_3-1&1
	\end{pmatrix}, \begin{pmatrix}
	1&t_1&s_1-1 \\ -1&1+s_2&t_2 \\  \xi&-1&1
	\end{pmatrix}$, \\
	$\begin{pmatrix}
	1+s_1&&s_2-1 \\ -1&1+ \xi& \\ s_3&-1&1+s_4
	\end{pmatrix}, \begin{pmatrix}
	1&s_2&s_1-1 \\  \xi-1&1& \\ &s_4-1&1+s_3
	\end{pmatrix},$ \\ $ \begin{pmatrix}
	1+s_2&s_1&-1 \\ -1&1& \xi \\ s_4&s_3-1&1
	\end{pmatrix}$.
\end{enumerate}

%In case 2, it is easy to recognize that, for any fixed matrix $M$ in the set above, there always exists a matrix $N$ in the same set which is connected to $M$ via a permutation of the columns and changing sign to some rows of $N$. 
The above set can be partitioned into equivalence classes according to the relation:
\begin{equation}
N\sim M \Leftrightarrow \exists F\in S_4, F'\in\mathbb{GL}(3,\mathbb{R})\ :\ N=F'MF.
\end{equation}
By Remark \ref{rem1} the above equivalence relation preserves the orbits of solutions of the linear systems.
It is easy to check that there are five equivalence classes represented by the following matrices:\\
$M_1 = \begin{pmatrix}
1+s_1&&s_2-1 \\ -1&1+ \xi& \\ s_3&-1&1+s_4
\end{pmatrix},$ \\ $M_2 = \begin{pmatrix}
1+s_1&&t_1-1 \\ t_2-1&1+s_2& \\ &t_3-1&1+s_3
\end{pmatrix}$, \\
$M_3 = \begin{pmatrix}
1&t_1&s_1-1 \\ t_2-1&1+s_2& \\ s_3&-1&1+t_3
\end{pmatrix},$ \\ $M_4 = \begin{pmatrix}
1&t_1&s_1-1 \\ s_2-1&1&t_2 \\ t_3&s_3-1&1
\end{pmatrix}$, \\
$M_5 = \begin{pmatrix}
1& \xi&-1 \\ s_2-1&1&t_2 \\ t_3&-1&1+s_3
\end{pmatrix}$.\\
The solutions for $M_4,M_5$ are $\mathcal{O}_{S_4}(\v{v}_s)$, $\mathcal{O}_{S_4}(\v{v}_b)$ (that have been already excluded) and $\mathcal{O}_{S_4}(\v{v}_{tt})$, that will be treated in Sec.~\ref{app:ExclusionA4}.

The three remaining cases are given in the following:
\begin{itemize}
 \item For $M_1$ one has $\mathcal{O}_{S_4}(\v{v}_s)$, $\mathcal{O}_{S_4}(\v{v}_b)$, $\mathcal{O}_{S_4}(\v{v}_c)$, $\mathcal{O}_{S_4}(\v{v}_{tt})$, and $\mathcal{O}_{S_4}(\v{v}_1)$, $\mathcal{O}_{S_4}(\v{v}_2)$, with
 \[
 \v{v}_1= \begin{pmatrix} \alpha\\ \alpha \\ \beta \end{pmatrix},\ \v{v}_2 = \alpha \begin{pmatrix}  3 \\ 1 \\ 2 \end{pmatrix};
 \]
 The systems in the same equivalence class are connected by the permutations $F \in \{ R^{\pm 1}, C,CR^{\pm 1} \}$.
 \item For $M_2$ one has $\mathcal{O}_{S_4}(\v{v}_1)$ and $\mathcal{O}_{S_4}(\v{v}_3)$, with
 \[
 \v{v}_3 = \begin{pmatrix} 0 \\ \alpha \\ \beta \end{pmatrix};
 \]
 \item For $M_3$ one has $\mathcal{O}_{S_4}(\v{v}_s)$, $\mathcal{O}_{S_4}(\v{v}_b)$, $\mathcal{O}_{S_4}(\v{v}_c)$, $\mathcal{O}_{S_4}(\v{v}_to)$, $\mathcal{O}_{S_4}(\v{v}_1)$ and $\mathcal{O}_{S_4}(\v{v}_4)$, with
 \[
\v{v}_4 = \begin{pmatrix}  \alpha \\ \beta \\ \frac{\alpha + \beta}{2} \end{pmatrix} .
 \]
 The systems in the same equivalence class are connected by the permutations $F	\in \{R^2, C\}$.
\end{itemize}
We notice that $\v{v}_2$ is a particular case of $\v{v}_4$, then we can just treat the latter. On the other hand, the vectors in $\mathcal{O}_{S_4}(\v{v}_3)$ cannot be solutions for $M_i$ with $i\neq 2$, otherwise the orbit is reduced to 
$\mathcal{O}_{S_4}(\v{v}_s)$, or $\mathcal{O}_{S_4}(\v{v}_c)$, or $\mathcal{O}_{S_4}(\v{v}_{to})$, which are ruled out. The remaining case of $\mathcal{O}_{S_4}(\v{v}_3)$ can be then excluded via the same analysis of case 2 in the previous section.

We end up with $\mathcal{O}_{S_4}(\v{v}_1),\mathcal{O}_{S_4}(\v{v}_4)$. We observe that imposing that $X_2\v{v}_1$ is a solutions for $M_1,M_2, M_3$ gives rise to $\mathcal{O}_{S_4}(\v{v}_s)$, $\mathcal{O}_{S_4}(\v{v}_b)$, $\mathcal{O}_{S_4}(\v{v}_c)$. As for $\mathcal{O}_{S_4}(\v{v}_4)$, imposing that $X_2\v{v}_4$ is a solution leads to $\mathcal{O}_{S_4}(\v{v}_{tt})$, $\mathcal{O}_{S_4}(\v{v}_{to})$, and $\mathcal{O}_{S_4}(\v{v}_5)$ with $\v{v}_5 = \alpha(5,3,1)^T$. However, it's easy to verify that $(I-X_2R)\v{v}_5$ is uniquely determined as sum of elements of $\{ \v{0}, \mathcal{O}_{S_4}(\pm \v{v}_5)\}$, leading us to exclude this last case by virtue of Prop.~\ref{P3}.

Finally, as anticipated at the beginning of Sec.~\ref{app:AS_4}, we can exclude $(I-R)\v{v} = \pm A\v{v}$ for $A\in L$ and $L\cong A_4,S_4$: by direct inspection of the representation matrices of $S_4$, it turns out that this condition leads to $\mathcal{O}_{S_4}(\v{v}_c)$.

\subsection{Exclusion of the truncated tetrahedron}\label{app:ExclusionA4}
In this section we make use of the three-dimensional irreducible representation of $A_4$ provided in Subsec. \ref{app:A_4} in order to exclude, by means of the unitarity conditions, the graph whose primitive cell is the set of vertices of the truncated tetrahedron. This also excludes the case where the inverses are contained in $S_+$. For notation convenience, we will use the Pauli matrices notation $X\coloneqq X_1$,  $Y\coloneqq X_2$, and  $Z\coloneqq X_3$, and use the vector $\v{w}_{tt}= \alpha(1,1,3)^T$ instead of $\v{v}_{tt}$ as a representative of the orbit  $\mathcal{O}_{A_4}(\v{v}_{tt})$. In the following we will also denote the elements $G\v{w}_{tt}$ (for $G\in A_4$) with the shorthand $G$.

Let $U$ be a faithful unitary and (generally projective) representation of $A_4$ in $\mathbb{SU}(2)$. We will denote the transition matrices as
\begin{equation}\label{ApmG}
A_{\pm G} := U_G A_{\pm I} U_G^\dagger,
\end{equation}
with $G\in A_4$. From the unitarity conditions \eqref{eq:condunit}, choosing $\v{h}''=2\v{w}_{tt}$, one derives the form
\begin{equation}\label{Apm}
A_{\pm I} \coloneqq \alpha_{\pm} V \ket{\pm}\bra{\pm},
\end{equation}
with $\lbrace \ket{+}, \ket{-} \rbrace$ orthonormal basis, $\alpha_{\pm} > 0$ and $V$ unitary. Consider the following unitarity conditions:
\[
A_IA_{W}^\dagger + A_{-W}A_{-I}^\dagger = 0,\quad W=X,Y.
\]
By multiplication on the right by $A_{I}$ we obtain
\[
A_IA_W^\dagger A_I = 0,
\]
implying that $U_W$ must be antidiagonal in the $\lbrace \ket{+}, \ket{-} \rbrace$ basis or in $\lbrace V\ket{+}, V\ket{-} \rbrace$. On the other hand, from
\[
A_IA_{-R}^\dagger + A_{R}A_{-I}^\dagger = 0,
\]
one gets
\begin{equation}\label{AARA}
A_{-I}A_{R}^\dagger A_{-I} = 0,
\end{equation}
meaning that $U_R$ must be diagonal in $\lbrace \ket{+}, \ket{-} \rbrace$ or $\lbrace V\ket{+}, V\ket{-} \rbrace$.

Let us now suppose that $U_X$ is antidiagonal in $\lbrace \ket{+}, \ket{-} \rbrace$ and $U_Y$ antidiagonal in $\lbrace V\ket{+}, V\ket{-} \rbrace$ (or viceversa): then, since
\begin{equation}\label{eq:XYZR}
U_RU_XU_R^\dagger = s_1 U_Y,\ U_RU_YU_R^\dagger = s_2 U_Z,
\end{equation}
($s_1,s_2$ arbitrary signs) all of the $U_G$ for $G=X,Y,Z$ would be antidiagonal in one of the two bases, but this violates the algebra of $D_2 \equiv \lbrace I,X,Y,Z \rbrace$ in $A_4$. Accordingly, choosing the $\lbrace \ket{+}, \ket{-} \rbrace$ basis and imposing
\begin{equation}
U_XU_Y = t_1 U_YU_X = t_2 U_Z,\ U_{G}^2 = t_3I \label{D2alg}
\end{equation}
(for $G=X,Y,Z$ and $t_1,t_2,t_3$ arbitrary signs), it is easy to see that up to a change of basis we can always take:
\[
U_G = i\sigma_G,\ G=X,Y,Z
\]
with $\ket{+}, \ket{-}$ eigenvectors of $\sigma_Z$. This implies that, in order to satisfy \eqref{eq:XYZR}, $U_R$ cannot have vanishing elements in $\lbrace \ket{+}, \ket{-} \rbrace$ and then by Eq.\eqref{AARA} it must be diagonal in $\lbrace V\ket{+}, V\ket{-} \rbrace$. Consequently we must have:
\begin{equation}
U_R \coloneqq VDV^\dagger,\label{VDV}
\end{equation}
where $D = \operatorname{diag}(e^{i\epsilon},e^{-i\epsilon})$ in $\lbrace \ket{+}, \ket{-} \rbrace$ and $e^{3i\epsilon}$ is a sign. As a consequence, using conditions \eqref{D2alg} one sees that the $U_X,U_Y,U_Z$ cannot have vanishing elements in $\lbrace V\ket{+}, V\ket{-} \rbrace$. This in turn implies, by Eq. \eqref{eq:XYZR}, that $V$ cannot have vanishing elements in $\lbrace \ket{+}, \ket{-} \rbrace$.

Let us now pose
\[
V = \begin{pmatrix}
\rho e^{i\theta} & \tau e^{i\varphi} \\
-\tau e^{-i\varphi} & \rho e^{-i\theta}
\end{pmatrix}\ :\ \rho , \tau >0,\ \rho^2+\tau^2 =1 .
\]
Multiplying on the left by $A_{-I}^\dag$ the following unitarity condition
\[
A_IA_{RX}^\dagger + A_{-RX}A_{-I}^\dagger + A_R A_Y^\dagger + A_{-Y}A_{-R}^\dagger = 0,
\]
and reminding that  $A_{-I}^\dag A_R A_Y^\dagger=0$ by Eq. (\ref{VDV}), one has
\[
A_{-I}^\dagger A_{-RX}A_{-I}^\dagger + A_{-I}^\dagger A_{-Y}A_{-R}^\dagger = 0 .
\]
Now, substituting Eq. (\ref{Apm}), and using definition (\ref{ApmG}), the nonvanishing matrix element of the previous identity in the basis $\lbrace \ket{+}, \ket{-} \rbrace$ is
\begin{equation*}
\begin{aligned}
& \bra{-}V^\dagger U_{RX} V\ket{-} \bra{-}U_{RX}^\dagger \ket{-} = \\ 
&= -e^{i \epsilon} \bra{-}V^\dagger U_Y V\ket{-} \bra{-} U_Y^\dagger U_R \ket{-} .
\end{aligned}
\end{equation*}
Recalling the form of $V$ given above and using the fact that $U_{RX} = t'U_RU_X$ ($t'$ a sign) and that $U_R$ cannot have vanishing elements in the basis $\lbrace \ket{+}, \ket{-} \rbrace$, for the previous equation we finally obtain
\[
\cos (\theta_1+\varphi_1) = - i \sin (\theta_1+\varphi_1) e^{2i\epsilon},
\]
which has no solution.

\subsection{Exclusion of $\mathbb{Z}_n$, $D_{n}$, $\mathbb{Z}_n\times \mathbb{Z}_2$ and $D_{n}\times \mathbb{Z}_2$, with $n=3,4,6$}\label{app:Z_n}
The aim of the present section is: 1) to construct the real, orthogonal and three-dimensional faithful representations of the groups $L \in \lbrace \mathbb{Z}_n, D_{n}, \mathbb{Z}_n\times \mathbb{Z}_2, D_{n}\times \mathbb{Z}_2 \left. | \right. n=3,4,6  \rbrace$, and 2) to exclude all the graphs arising from $L$ by means of the unitarity conditions.

By the classification theorem for real matrices of finite order given in Ref. \cite{koo2003classification}, any matrix in $\mathbb{O}(3)$ of order $n$ is similar to one of the form

\begin{equation*}
R_{\theta, s} \coloneqq \begin{pmatrix}
\cos \theta & -\sin \theta & 0 \\
\sin \theta & \cos \theta & 0 \\
0 & 0 & s
\end{pmatrix},
\end{equation*}
with $\theta = \frac{2z\pi}{n}$, $z$ integer and $s$ a sign. The matrices $R_{\theta, s}$ represent the generators for the subgroups of order $n=3,4,6$ in $L$. We can generate the orbits of $L$ starting from the generic vector (up to a rotation around the $z$-axis) given by $\v{v}_1= (1,0,h)^T$. It is easy to show that the only matrices in $\mathbb{O}(3)$ of order 2 commuting with $R_{\theta ,s }$ for all $\theta$ and $s$ are $R_{0,t}$ and $R_{\pi , t}$: they represent the generators of $L/\mathbb{Z}_n$ for $L\cong\mathbb{Z}_n\times\mathbb{Z}_2$ or  $L/D_n$ for $L\cong D_n\times\mathbb{Z}_2$ . On the other hand, the involutions
\begin{equation*}
S_{\varphi, r} \coloneqq \begin{pmatrix}
\cos \varphi & \sin \varphi & 0 \\
\sin \varphi & -\cos \varphi & 0 \\
0 & 0 & r
\end{pmatrix}
\end{equation*}
are the only ones such that $S_{\varphi, r} R_{\theta, s} S_{\varphi, r}^{-1} = S_{\varphi, r} R_{\theta, s} S_{\varphi, r} = R_{\theta, s}^{-1}$. This implies that the $S_{\varphi, r}$ represent the generators for the subgroups of reflections when $L$ is a dihedral group. Therefore, in general, the elements of $\mathcal{O}_L(\v{v}_1)$ lie on the two circumferences which are parallel to the $xy$-plane at heights $z=\pm h$.

In order to solve the unitarity conditions, it is necessary to determine the paths with length 2 constructed by elements in $\{\v{0}\}\cup \mathcal{O}_L(\v{v}_1)$: by the above analysis, the problem is reduced to a two-dimensional problem, since the form of the vectors in $\mathcal{O}_L(\v{v}_1)$ is $\v{v}_i = (x_i,y_i,\pm h)^T  \coloneqq (\cos \chi_i ,\sin \chi_i,\pm h)^T $. Accordingly, it is easy to see that
\begin{equation*}
\v{v}_i\pm \v{v}_j = s\v{v}_l + t\v{v}_p,\quad \v{v}_i,\v{v}_j,\v{v}_l,\v{v}_p \neq \v{0},\;\text{ ($s,t$ signs)}
\end{equation*}
implies $(x_i,y_i) = s(x_l,y_l)$ or $(x_i,y_i) = t(x_p,y_p)$.

\textbf{Case $\mathbf{n=4}$.} There are at least two inequivalent orthogonal representations of $L \in \lbrace \mathbb{Z}_4, D_{4}, \mathbb{Z}_4\times \mathbb{Z}_2, D_{4}\times \mathbb{Z}_2 \rbrace$, since the element of order 4 can be either represented by $R_{\frac{\pi}{2},-}$ or $R_{\frac{\pi}{2},+}$. We shall now analyse the two different cases.

$R_{\frac{\pi}{2},-}$ generates the four vectors
\begin{equation*}
\v{v}_1=\begin{pmatrix}
1 \\ 0 \\ h
\end{pmatrix}, \v{v}_2=\begin{pmatrix}
0 \\ 1 \\ -h
\end{pmatrix}, \v{v}_3=\begin{pmatrix}
-1 \\ 0 \\ h
\end{pmatrix}, \v{v}_4=\begin{pmatrix}
0 \\ -1 \\ -h
\end{pmatrix} .
\end{equation*}
The differences $\v{v}_i - \v{v}_j \neq 0$ $\forall i,j\in \{1,2,3,4 \}$ are uniquely determined as sums of elements of $\{ \v{0}, \mathcal{O}_L(\pm \v{v}_1)\}$. Accordingly, there is a cyclic subgroup of order 4 (i.e.~ $\mathbb{Z}_4$) whose orbit satisfies Eq. \eqref{e21} and thus, invoking Prop. \ref{prop1} (we remind that the representation $U$ must be faithful), we exclude the representation containing $R_{\frac{\pi}{2},-}$.

Taking now $R_{\frac{\pi}{2},+}$, the orbit is
\begin{equation*}
\v{v}_1=\begin{pmatrix}
1 \\ 0 \\ h
\end{pmatrix}, \v{v}_2=\begin{pmatrix}
0 \\ 1 \\ h
\end{pmatrix}, \v{v}_3=\begin{pmatrix}
-1 \\ 0 \\ h
\end{pmatrix}, \v{v}_4=\begin{pmatrix}
0 \\ -1 \\ h
\end{pmatrix}.
\end{equation*}
We have that the vectors
\begin{equation}\label{paths}
\v{v}_1 + \v{v}_2,\quad \v{v}_1 - \v{v}_3
\end{equation}
are uniquely determined as sum of elements of $\{ \v{0}, \mathcal{O}_L(\pm \v{v}_1)\}$. Let us denote with $R$ the matrix representing $R_{\frac{\pi}{2},+}$ in $\mathbb{SU}(2)$ and proceed as in Sec. \ref{app:ExclusionA4}. From now on in the present section we use the notation of Eq. \eqref{ApmG} and perform calculations in the $\left\lbrace \ket{+},\ket{-} \right\rbrace$ basis. Multiplying on the right by $A_{\v{v}_1}$ the unitarity conditions associated to the vectors in \eqref{paths}, we obtain
\begin{equation}\label{uu2}
A_{\v{v}_1}RA_{-\v{v}_1}^{\dagger}R^{\dagger}A_{\v{v}_1}=0,\ A_{\v{v}_1}R^2A_{\v{v}_1}^{\dagger}{R^2}^{\dagger}A_{\v{v}_1}=0.
\end{equation}
By the first of conditions \eqref{uu2}, up to a change of basis we can impose
\begin{equation*}
R = \begin{pmatrix}
\mu & 0 \\
0 & \mu^*
\end{pmatrix},\ R^4 = s I
\end{equation*}
($s$ arbitrary sign); using the second condition, it follows that
\begin{equation}\label{u2}
R^2 = \begin{pmatrix}
\mu^2 & 0 \\
0 & {\mu^*}^2
\end{pmatrix} = V \begin{pmatrix}
0 & \nu \\
-\nu^* & 0
\end{pmatrix}V^{\dagger},
\end{equation}
and thus necessarily $\mu^2 \neq {\mu^*}^2$. Consider now the unitarity condition
\[
A_{\v{v}_1}A_{\v{v}_2}^\dagger + A_{-\v{v}_2}A_{-\v{v}_1}^\dagger + A_{\v{v}_4} A_{\v{v}_3}^\dagger + A_{-\v{v}_3}A_{-\v{v}_4}^\dagger = 0.
\]
Multiplying the last equation by $A_{\v{v}_1}$ on the right and taking the adjoint we get~\footnote{One has $A_{-\v{v}_2}A_{-\v{v}_1}^\dag A_{\v{v}_1}=0$ since $A_{-\v{v}_1}^\dag A_{\v{v}_1}=0$, and $A_{\v{v}_4}A_{\v{v}_3}^\dag A_{\v{v}_1}=0$,  since $A_{\v{v}_3}^\dag A_{\v{v}_1}=R^2 A_{\v{v}_1}^\dag R^{2\dag}A_{\v{v}_1}$ and  $ A_{\v{v}_1}^\dag R^2 A_{\v{v}_1}=\alpha_+^2\ket{+}\bra{+}V^\dag R^2V\ket{+}\bra{+}$, and by Eq. (\ref{u2}) $\bra{+}V^\dag R^2V\ket{+}=0$.
}

\[
A_{\v{v}_{1}}^\dagger A_{\v{v}_2}A_{\v{v}_{1}}^\dagger + A_{\v{v}_{1}}^\dagger A_{-\v{v}_4}A_{-\v{v}_3}^\dagger = 0,
\]
which amounts to
\begin{equation}\label{eq:elcond}
\begin{aligned}&\frac{\alpha_+^2}{\alpha_-^2}\bra{+}R^\dag\ket{+}\bra{+}V^\dagger R V\ket{+}=\\
=& -\nu^*\bra{-}R^\dagger\ket{-} \bra{+}V^\dagger {R^3}V\ket{-}.
\end{aligned}
\end{equation}
Posing now
\begin{equation*}
V^\dagger R V = \begin{pmatrix}
a & b \\
-b^* & a^*
\end{pmatrix},
\end{equation*}
with $a,b \neq 0$ since otherwise $V^\dagger R^2 V$ cannot be anti-diagonal (see \eqref{u2}), we have that
\begin{equation*}
V^\dagger {R^3} V = (V^\dagger {R} V)(V^\dagger {R^2} V) =
\begin{pmatrix}
a & b \\
-b^* & a^*
\end{pmatrix}
\begin{pmatrix}
0 & \nu \\
-\nu^* & 0
\end{pmatrix}.
\end{equation*}
Accordingly, Eq. \eqref{eq:elcond} leads to
\begin{equation*}
\frac{\alpha_+^2}{\alpha_-^2} = -\mu^2,
\end{equation*}
which is impossible, since $\mu^2 \neq {\mu^*}^2$.

\textbf{Cases $\mathbf{n=3,6}$.} The representations of $L \in \lbrace \mathbb{Z}_n, D_{n}, \mathbb{Z}_n\times \mathbb{Z}_2, D_{n}\times \mathbb{Z}_2 \left. | \right. n=3,6  \rbrace$ must contain $R_{\frac{2\pi}{3}, +}$, which generates a subgroup $K$ isomorphic to $\mathbb{Z}_3$: $\mathcal{O}_K(\v{v}_1)$ is given by the following vectors:
\begin{equation*}
\v{v}_l = \begin{pmatrix}
\cos \frac{2\pi}{3}(l-1) \\ \sin \frac{2\pi}{3}(l-1)  \\ h
\end{pmatrix},\ l\in \{1,2,3\}.
\end{equation*}
We denote the representation matrix of $R_{\frac{2\pi}{3}, +}$ in $\mathbb{SU}(2)$ with $U_{\frac{2\pi}{3}}$.

If $\v{v}_1-\v{v}_2$ is uniquely determined as sum of elements of $\{ \v{0}, \mathcal{O}_L(\pm \v{v}_1)\}$ (a particular case is given by the condition $h=0$), we can exclude this case by Prop. \ref{P3}. Let us then suppose that $\v{v}_1-\v{v}_2$ is not uniquely determined as sum of elements of $\{ \v{0}, \mathcal{O}_L(\pm \v{v}_1)\}$ (in particular $h\neq 0$) . Then, by the above analysis, $\mathcal{O}_L(\v{v}_1)$ must contain
\begin{equation*}
\v{v}_{l} = \begin{pmatrix}
-\cos \frac{2\pi}{3}(l-1) \\ -\sin \frac{2\pi}{3}(l-1)  \\ h
\end{pmatrix},\ l\in \{4,5,6\}
\end{equation*}
(such that $\v{v}_1-\v{v}_2 = \v{v}_{5}-\v{v}_{4}$). Again, via the above arguments on the representations of $L$, is easy to see that $\v{v}_1 + \v{v}_2$ is uniquely determined as sum of elements of $\{ \v{0}, \mathcal{O}_L(\pm \v{v}_1)\}$. Then, from condition
\begin{equation*}
A_{\v{v}_1}A_{-\v{v}_2}^\dag+A_{\v{v}_2}A_{-\v{v}_1}^\dag=0,
\end{equation*}
by multiplying on the right by $A_{\v{v}_1}$, we get
\begin{equation*}
A_{\v{v}_1}A_{-\v{v}_2}^\dagger A_{\v{v}_1} = 0.
\end{equation*}
Up to a change of basis $U_{\frac{2\pi}{3}} = \operatorname{diag} (e^{i\epsilon}, e^{-i\epsilon})$ holds with $e^{3i\epsilon}=\pm1$, and $\epsilon\not\in\{0,\pi\}.$
\begin{comment}
we also have the condition
\begin{equation*}
A_{\v{v}_1}A_{\v{v}_2}^\dagger A_{\v{v}_1} = 0,
\end{equation*}
which amounts to $V^\dagger U_{\frac{2\pi}{3}} V$ anti-diagonal, that is absurd since it must also be $V^\dagger U_{\frac{2\pi}{3}}^3 V = \pm I$.
\end{comment}
Let $U_\pi$ represent the element of $L$ mapping $\v{v}_1$ to $\v{v}_{4}$. This element is an involution and there are only two cases (by inspection of the groups $L$ here considered)
\begin{equation}
U_\pi U_{\frac{2\pi}{3}} U_\pi^\dagger=
\begin{cases}
&s U_{\frac{2\pi}{3}} \\
&\\
&s' U_{\frac{2\pi}{3}}^\dagger ,
\end{cases}
\end{equation}
($s,s'$ signs). Recalling that the representation $U \subset \mathbb{SU}(2)$ is faithful and $U_{\frac{2\pi}{3}}^3 = tI$ ($t$ a sign), it is easy to verify that the previous two conditions on $U_{\frac{2\pi}{3}},U_\pi$ are satisfied respectively only if
\begin{enumerate}
\item $U_\pi$ is diagonal; \label{udiag}
\item $U_\pi$ is anti-diagonal. \label{uanti}
\end{enumerate}
Multiplying by $A_{\v{v}_1}$ on the right the unitarity condition associated to the difference $\v{v}_1-\v{v}_{4}$
\begin{equation*}
A_{\v{v}_1}A_{\v{v}_4}^\dag+A_{-\v{v}_4}A_{-\v{v}_1}^\dag=0,
\end{equation*}
one also gets
\begin{equation*}
A_{\v{v}_1}A_{\v{v}_4}^\dagger A_{\v{v}_1} = 0,
\end{equation*}
namely either $A_{\v{v}_4}^\dagger A_{\v{v}_1}=0$ or $A_{\v{v}_1}A_{\v{v}_4}^\dagger=0$. This implies that a) $V^\dagger U_{\pi} V$ is anti-diagonal or b) $U_{\pi}$ is anti-diagonal. In case a), multiplying by $A_{\v{v}_1}$ on the right the unitarity condition
\begin{equation*}
A_{\v{v}_1}A_{\v{v}_2}^\dagger + A_{-\v{v}_2}A_{-\v{v}_1}^\dagger + A_{\v{v}_5} A_{\v{v}_4}^\dagger + A_{-\v{v}_4}A_{-\v{v}_5}^\dagger = 0,
\end{equation*}
it follows that 
\begin{equation}\label{eq:unitZ6first}
A_{\v{v}_1}A_{\v{v}_2}^\dagger A_{\v{v}_1} + A_{-\v{v}_4}A_{-\v{v}_5}^\dagger A_{\v{v}_1} = 0 ;
\end{equation}
in case b) multiplying by $A_{\v{v}_1}^\dagger$ on the right the unitarity condition
\begin{equation}
A_{\v{v}_1}^\dagger A_{\v{v}_2} + A_{-\v{v}_2}^\dagger A_{-\v{v}_1} + A_{\v{v}_5}^\dagger A_{\v{v}_4} + A_{-\v{v}_4}^\dagger A_{-\v{v}_5} = 0,
\end{equation}
and taking the adjoint, it follows that 
\begin{equation}\label{eq:unitZ6second}
A_{\v{v}_1} A_{\v{v}_2}^\dagger A_{\v{v}_1} + A_{\v{1}} A_{-\v{v}_5}^\dagger A_{-\v{v}_4} = 0.
\end{equation}
Let us now pose
\begin{equation*}
V^\dagger U_{\frac{2\pi}{3}} V = \begin{pmatrix}
 a  &  b \\
 -b^* & a^*
\end{pmatrix},
\end{equation*}
where $a \neq 0$ since $U_{\frac{2\pi}{3}}^3 = tI$. In case a), from \eqref{eq:unitZ6first} one then has
\begin{equation*}
\frac{\alpha_+^2}{\alpha_-^2} e^{i\epsilon} = -\bra{-}U_\pi^\dagger U_{\frac{2\pi}{3}}U_\pi\ket{-},
\end{equation*}
which cannot be satisfied neither in case \ref{udiag} nor in case \ref{uanti}. On the other hand in case b), from \eqref{eq:unitZ6second} one has
\begin{equation*}
\frac{\alpha_+^2}{\alpha_-^2} a^*e^{i\epsilon} = -e^{-i\epsilon}\bra{-}V^\dagger U_\pi^\dagger U_{\frac{2\pi}{3}}^\dagger U_\pi V\ket{-},
\end{equation*}
and being $U_\pi$ anti-diagonal, one has
\begin{equation*}
\frac{\alpha_+^2}{\alpha_-^2} a^*e^{2i\epsilon} = - \bra{-}V^\dagger U_{\frac{2\pi}{3}} V\ket{-} = -a^*,
\end{equation*}
which is impossible, since $e^{3i\epsilon}=\pm1$, for $\epsilon\not\in\{0,\pi\}$.

\subsection{Remaining presentations arising from $\mathbb{Z}_2$, $D_2$ and $D_2 \times \mathbb{Z}_2$}\label{app:Z_2}
By the argument of Sec. \ref{app:Z_n}, any matrix of order 2 in $\mathbb{O}(3)$ is similar to
\begin{equation*}
M_{s,t} \coloneqq \begin{pmatrix}s & 0 & 0 \\
0 & s & 0 \\
0 & 0 & t
\end{pmatrix},
\end{equation*}
with $s,t$ signs. Accordingly, up to conjugation, any three-dimensional orthogonal representation of a group $L\in \left\lbrace\mathbb{Z}_2,D_2,D_2 \times \mathbb{Z}_2\right\rbrace$ contains $M_{s,t}$. If $s \neq t$, any matrix $N$ of order 2 in $\mathbb{O}(3)$ commuting with $M_{s,t}$ is either $M_{s',t'}$, or of the form:
\begin{equation*}
N = \begin{pmatrix}
\cos \varphi & \sin \varphi & 0 \\
\sin \varphi & -\cos \varphi & 0 \\
0 & 0 & r
\end{pmatrix},
\end{equation*}
with $r$ a sign. Being the two dimensional block a reflection matrix, there exists a similarity transformation which maps it to $\pm \sigma_{z}$ (and leaving $M_{s,t}$ invariant). Thus the real, orthogonal and three-dimensional faithful representations of the groups here considered contain just $M_{s,t}$ and
\begin{equation*}
N_{r_1, r_2} \coloneqq \begin{pmatrix}
r_1 & 0 & 0 \\
0 & -r_1 & 0 \\
0 & 0 & r_2
\end{pmatrix}.
\end{equation*}
The problem reduces to combine signs in $M_{s,t},N_{r_1,r_2}$ to give rise to faithful representations of $L$. It is easy to check that they give rise to the integer lattice, the square lattice or the BCC lattice (one can include the inverses or not). Nevertheless, there are two ways of providing a minimal generating set (namely such that $S_+\neq S_-$) for $\mathbb{Z}^3$ and whose Cayley graph is associated with the BCC lattice. Such presentations are both generated by $D_2$: one is made with the vertices of a tetrahedron; the second one corresponds to the vertices given by the following vectors
\begin{equation*}
\v{v}_0 = \begin{pmatrix} 1 \\ 1 \\ h \end{pmatrix}, \v{v}_1 = \begin{pmatrix} -1 \\ -1 \\ h \end{pmatrix},
\v{v}_2 = \begin{pmatrix} -1 \\ 1 \\ h \end{pmatrix},
\v{v}_3 = \begin{pmatrix} 1 \\ -1 \\ h \end{pmatrix}.
\end{equation*}
We notice that excluding this solution allows us to exclude the case including the inverses, namely $S_+=S_-$.

From the unitarity conditions one has:
\begin{align}
&A_{\v{v}_0}A_{\v{v}_1}^\dagger A_{\v{v}_0} = 0, \label{eq:unit1}  \\ 
&A_{\v{v}_0}A_{-\v{v}_i}^\dag A_{\v{v}_0}=0,\quad i=2,3, \label{eq:unit2} \\
&A_{\v{v}_0}A_{-\v{v}_1}^\dagger + A_{\v{v}_1}A_{-\v{v}_0}^\dagger + A_{\v{v}_2}A_{-\v{v}_3}^\dagger + A_{\v{v}_3}A_{-\v{v}_2}^\dagger = 0. \label{eq:unit3} %\\
%A_{\v{v}_1}^\dagger A_{Y,Z} + A_{-Y,-Z}^\dagger A_{-\v{v}_1} + A_{Z,Y}^\dagger A_{X} + A_{-X}^\dagger A_{-Z,-Y} = 0 \label{eq:unit4}
\end{align}
From \eqref{eq:unit1} and the form of Eqs. \eqref{ApmG},\eqref{Apm} for the transition matrices, we get $U_1 = i\sigma_1$ (we use the equivalent notation for Pauli matrices: $\sigma_0:=I$,
 $\sigma_1:=\sigma_X$, $\sigma_2:=\sigma_Y$, $\sigma_3:=\sigma_Z$),   
up to a change of basis; from \eqref{eq:unit2} we end up with the two cases:
\begin{enumerate}
	\item $A_{-\v{v}_2}^\dagger A_{\v{v}_0} = A_{-\v{v}_3}^\dagger A_{\v{v}_0} = 0$,
	\item $A_{-\v{v}_2}^\dagger A_{\v{v}_0} = A_{\v{v}_0} A_{-\v{v}_3}^\dagger = 0$,
\end{enumerate}
since $ A_{\v{v}_0}A_{-\v{v}_2}^\dagger =  A_{\v{v}_0}A_{-\v{v}_3}^\dagger = 0$ is forbidden in order to respect the $D_2$ algebra, while the case $A_{-\v{v}_3}^\dagger A_{\v{v}_0} = A_{\v{v}_0} A_{-\v{v}_2}^\dagger = 0$ is accounted by the symmetry of the unitarity conditions under the exchange $2 \leftrightarrow 3$. In case 1, the condition is incompatible with a faithful representation of $D_2$ in $\mathbb{SU}(2)$. In case 2, we have $U_G = i\sigma_G$ and $U_{2} = VDV^\dagger$ with $D$ diagonal, implying $U_{2} = s V (i\sigma_3) V^\dagger$ ($s$ a sign). Then, up to a global sign, one has
\begin{equation*}
V = \frac{1}{\sqrt{2}}
\begin{pmatrix}
 i  &  s\\
-s & - i 
\end{pmatrix},
\end{equation*}
and from \eqref{eq:unit3} we obtain
\begin{equation*}
A_{\v{v}_0}A_{-\v{v}_1}^\dagger A_{\v{v}_0} + A_{\v{v}_2}A_{-\v{v}_3}^\dagger A_{\v{v}_0} = 0,
\end{equation*}
which, using the form of Eqs. \eqref{ApmG},\eqref{Apm} for the transition matrices along with the previous results, leads to $-2\alpha_+^2\alpha_-=0$, contradicting the assumption $\alpha_\pm\neq 0$.

\bibliography{Invitation}

\end{document}